\documentclass[10pt]{amsart}

\usepackage{graphicx}
\usepackage{epstopdf}
\usepackage{cite}
\usepackage{latexsym}
\usepackage{amsmath,amssymb,amsfonts,amscd}
\usepackage{mathrsfs}
\usepackage{arydshln}
\usepackage{url}
\usepackage{hyperref}
\usepackage{balance}
\usepackage{caption}
\usepackage{framed}
\usepackage{enumitem}
\usepackage{subfigure}
\usepackage{algorithm}
\usepackage{algorithmic}
\makeatletter
\newcommand{\setalglineno}[1]{%
  \setcounter{ALC@line}{\numexpr#1-1}}
\makeatother
\usepackage{cleveref}
\usepackage{float}

\newtheorem{theorem}{Theorem}[section]

\newtheorem{proposition}[theorem]{Proposition}

\newtheorem{definition}[theorem]{Definition}
\theoremstyle{definition}

\theoremstyle{remark}

\newcommand{\CC}{\mathbb{C}}
\newcommand{\RR}{\mathbb{R}}
\newcommand{\PP}{\mathbb{P}}

\newcommand{\mb}[1][]{\mathbf}
\newcommand{\m}[1]{\mb{m_{#1}}}
\newcommand{\x}{\mb{x}}

\newcommand{\bs}[1][]{\boldsymbol}
\newcommand{\bt}{\boldsymbol{\tau}}
\newcommand{\eqn}{\begin{eqnarray}}
\newcommand{\feqn}{\end{eqnarray}}

\newcommand{\PM}[1]{\langle #1 \rangle}

\newcommand{\unaryminus}{{\tiny\scalebox{0.5}[1.0]{\( - \)}}}

\begin{document}

\title[Statistical Model of Source Localization based on Range Differences]{On the Statistical Model of Source Localization based on Range Difference Measurements}

\author[M. Compagnoni, R. Notari, F. Antonacci, A. Sarti]{Marco Compagnoni, Roberto Notari, Fabio Antonacci, Augusto Sarti}

\begin{abstract}
In this work we study the statistical model of source localization based on Range Difference measurements. We investigate the case of planar localization of a source using a minimal configuration of three non aligned receivers.
Our analysis is based on a previous work of the same authors concerning the localization in a noiseless scenario. As the set of feasible measurements is a semialgebraic variety, this investigation makes use of techniques from Algebraic Statistics and Information Geometry.\\[2mm]
\\[2mm]
\end{abstract}

\address{Dipartimento di Matematica, Politecnico di Milano, I-20133 Milano, Italy}
\email{marco.compagnoni@polimi.it}
\email{roberto.notari@polimi.it}
\address{Dipartimento di Elettronica, Informazione e Bioingegneria, Politecnico di Milano, I-20133 Milano, Italy}
\email{fabio.antonacci@polimi.it}
\email{augusto.sarti@polimi.it}

\thanks{The authors are grateful to Alessandra Guglielmi for the useful discussions and suggestions during the preparation of this work.}

\date{\today}

\maketitle


\section{Introduction}\label{sc:intro}
Source localization from the analysis of the signals captured by multiple sensors is a classical research theme in science and engineering. Among the early studies on this subject (dating back to World War II) is the analysis of the two-dimensional LOng RAnge Navigation (LORAN) radio positioning system. LORAN was based on the measurements of time differences of arrival (TDOAs) of synchronized radio signals originated from three distinct known emitters. The method needed hyperbolic charts for determining the position of the receiver \cite{Getting1993}. Since then, there has been a proliferation of areas of applications where source localization plays a fundamental role. Among them are radar and sonar technologies; wireless sensor networks, the Global Positioning System (GPS); and robotics.
LORAN, in particular, is an example of localization technology based on Range Differences (RD), or pseudoranges. This technique is characterized by:
\begin{itemize}
\item
a point $\x$, whose location we want to find;
\item
a set $\{\m{0},\dots,\m{n}\}$ of points placed at known positions;
\item
the RDs of the signals emitted by $\{\m{i},\m{j}\}$ measured at $\x$ as experimental data.
\end{itemize}
RD--based localization is particularly popular in audio signal processing, where pseudoranges are usually computed from the measurements of the TDOAs between calibrated and synchronized microphones \cite{Benesty2004,Chen2002,Hu2010,Knapp1976,Ianniello1982}. As RDs and TDOAs are simply proportional to each other (given the sound propagation speed), we will treat (with a slight abuse of notation) RD and TDOA as synonymous throughout this manuscript. In other context, e.g. remote sensing, radar and GPS \cite{Siouris1993,Yimin2008}, RDs can again be derived from TDOAs \cite{yang2009efficient} or through other approaches such as energy measurements \cite{li2003energy}.

In the signal processing literature, we can find various examples of analysis of localization models based on numerical simulations. For example, a study of the TDOA--based localization for a minimal configuration of sensors (three receivers coplanar with the source) can be found in \cite{Spencer2007,Spencer2010}. Therein, the author makes use of the concept of TDOA space and offers a first description of the feasible set of TDOAs. Given the importance of the topic, in \cite{Bestagini2013,Compagnoni2013b,Compagnoni2013a} we offered a systematic and comprehensive analytic investigation of the mathematical models behind TDOA--based source localization. Using algebraic and geometric tools, we studied in details the deterministic model for the minimal TDOA--based localization:
\begin{itemize}
\item
we defined the TDOA map from the physical plane of source location to the space of TDOA measurements, which completely encodes the noiseless localization model;
\item
we described the image of the TDOA map, i.e. the set of feasible noiseless measurements;
\item
we studied the invertibility of the map and, consequently, the existence and uniqueness of the source for any given set of measurements.
\end{itemize}
As a confirmation of the importance of these topics for applications we can cite \cite{Alameda2014}, which describes the use of TDOA measurements set for TDOA estimation. Similar works have been carried out for different kinds of measurements as well. In particular, in \cite{Compagnoni2016b} the authors describe Range--based localization models, while Directions Of Arrival measurements are considered in \cite{schicho2015ambiguities}.
It is well known, however, that in real world scenarios localization techniques are sensitive to measurement noise. The sources of disturbance that tend to affect measurements in audio signal processing can be broadly classified into additive noise (due to time sampling, circuit noise and other physical phenomena) and outlier measurements (produced by reverberation or interfering sources). In order to deal with these problems, it is necessary to go further in the study of the models, and move from deterministic to statistical modeling. A first step towards this goal was taken in \cite{Compagnoni2016a}, where a denoising removal algorithm was proposed, based of the analysis presented in \cite{Compagnoni2013a}.
In \cite{Compagnoni2016a}, however, the authors could not exploit the full potential of the description of \cite{Compagnoni2013a} as the inherent complexity of the model called for a detailed study that would explicitly be devoted to addressing the problems of multiple localization and parameter estimation from a statistical standpoint. This is, in fact, the goal of this manuscript.

We will leverage on the results contained in \cite{Compagnoni2013a} to achieve the following goals:
\begin{enumerate}
\item
to study the statistical model behind TDOA-based localization for the minimal case of three receivers and one coplanar source. We will give particular care to the problem of ambiguity in localization;
\item
to provide an effective Maximum Likelihood localization technique that, given the range differences and the location of the sensors, computes the source location;
\item
to develop a technique that, given the sensor locations and an estimate of the measurement error magnitude, predicts the localization error covariance, as well as its bias.
\end{enumerate}

In order to attain the first goal, we leverage on the use of Information Geometry \cite{Amari2000}, which turns out to be particularly suitable for our geometric approach to the localization problem (see also \cite{Cheng2013} for its use in the context of Range--based localization). Moreover, Information Geometry allows us to apply the asymptotic theory of estimation for studying the accuracy of source localization, which is our third goal in the list. In the literature, the asymptotic estimation of the Root Mean Square Error (RMSE) and the bias are also among the goals of \cite{Ho2012,So2013}, although pursued with different tools. In this work we push the boundary a bit further: will focus on predicting the accuracy of the asymptotic estimation through the analysis of higher order statistics.

As far as the second goal is concerned, the fact that the Maximum Likelihood Estimation (MLE) is optimal from a statistical point of view is well known, as it attains the Cramer-Rao Lower Bound.
In the literature MLE algorithms are based on the maximization of the likelihood function, which depends on the coordinates of the source. Unfortunately, the nonlinearity and the non-convexity of the likelihood function make it quite difficult to formulate an effective solution, which is why other sub-optimal techniques are mostly used \cite{Hahn1973,Stoica1988,Schau1987,Huang2001,Beck2008,Schmidt1972,Compagnoni2012,Canclini2015}.
In our approach, the estimation is performed in the parameter space of the model. We exploit the knowledge on the geometry of the set of feasible measurements for obtaining a (quasi) closed-form solution of MLE. In our framework, MLE is equivalent to finding the solution of the geometric problem of projecting a point onto the set of feasible measurements, according to a suitable Euclidean structure defined on the measurements space. As proven in \cite{Compagnoni2013a}, the set of feasible measurements is a semi-algebraic variety, therefore our analysis naturally falls within the domain of Algebraic Statistics \cite{Draisma2013}.

Our rigorous analysis of localization in the minimal sensors configuration is particularly interesting for applications where one has some constraint on the amount of measurements. E.g. in GPS localization, where the number of available satellites is bounded and it is necessary to consider minimal information scenarios \cite{Abel1991,Awange2002,Bancroft1985,Chauffe1994,Coll2012,Coll2009,Grafarend2002,
Hoshen1996,Leva1995,Schmidt1972}. However, our contribution can be useful also in other fields, such as in audio signal processing and wireless sensor networks. Although in these contexts the number of sensors to handle is usually larger than the minimum, there are applications where it is convenient to focus on smaller subsets of them. For example, this is the case of robust estimation techniques like RANdom SAmple Consensus (RANSAC) algorithm \cite{Fischler1981}. Indeed, by considering few measurements at a time and combining the corresponding estimations, one can develop tests on the single measurements and identify the inliers and outliers in the dataset. We finally remark that the techniques that we develop in this manuscript are the basis also for the study of more general situations, with a greater number of sensors.

The paper is organized as follows. In Section \ref{sc:TDOAspace}, we recount the main results introduced in \cite{Compagnoni2013a} on the deterministic model for TDOA--based source localization in a minimal sensing scenario.
In Section \ref{sc:StatMod} we focus on accurately defining the statistical model. This is a rather delicate task, due to the difficulties that arise from localization ambiguities. In our approach, we choose to consider the model as a composition of four distinct curved exponential families, one for each region where the restriction of the TDOA map is a diffeomorphism between the physical and the measurements spaces.
In Section \ref{sc:sourceest} we address the MLE in the measurements space. As mentioned above, this is equivalent to studying the orthogonal projection of a point onto the set of feasible TDOA measurements.
Section \ref{sc:ASI} is devoted to studying the accuracy of source localization via MLE. Our analysis is based on asymptotic statistical inference through the approach of Information Geometry. In particular, we obtain an analytic form for the mean square error and the bias of the MLE. Moreover, in Subsection \ref{sc:assASI} we propose a method for evaluating the reliability of the asymptotical inference, based on higher-order statistics.
Section \ref{sec:validation} looks at the problem from a practical standpoint. In Subsection \ref{sec:MLEfinale} we explicitly describe the MLE algorithm for each one of the four models defined in Section \ref{sc:StatMod}. In Subsection \ref{sc:simulations} we conduct a simulation campaign, aimed at validating our algorithms and conducting an asymptotic error analysis. In Subsection \ref{sec:blindMLE} we then give indications on the source localization problem in a real scenario, in which we don't not know in advance which model to use.
In Section \ref{sec:conclusion} we briefly discuss the potential impact of this work and draw some conclusions. Finally, in Appendix \ref{app:evoluta} we include the code for computing the Cartesian equation of the Mahalanobis degree discriminant of an ellipse, which has a role in the computation of the MLE.


\section{The TDOA space and the deterministic model}\label{sc:TDOAspace}

The TDOA space and the TDOA maps were introduced in \cite{Spencer2007,Compagnoni2013a} for the analysis of TDOA--based source localization with a minimal configuration of three receivers in two dimensions. In this section, we briefly go over the main results of
\cite{Compagnoni2013a,Compagnoni2013b}, using the same tools and notations. In order to simplify matters, we only describe the case in which the receivers are not collinear. The interested reader can develop a similar statistical analysis for the case of aligned sensors starting from \cite{Compagnoni2013a}.

One of the main mathematical tools used in \cite{Compagnoni2013a} is the exterior algebra formalism over the three dimensional Minkowski vector space $\RR^{2,1},$ which roughly corresponds to the product of the Euclidean physical plane times the real line containing the TDOAs. Actually, this instrument is very useful for handling the equations involved in the localization problem. We refer to Appendix A of \cite{Compagnoni2013a} for an introduction to the subject. However, in this manuscript it is sufficient to use the exterior algebra formalism over the Euclidean vector space $\RR^2.$ For the convenience of the reader, here we summarize the main facts for this particular case. 

Let $ V $ be a $ 2$--dimensional Euclidean vector space and let $ B = \{\mb{e_1}, \mb{e_2}\}$ be an orthonormal basis. With a slight abuse of notation, we identify a vector $\mb{v}=v_1\mb{e_1}+v_2\mb{e_2}$ with its coordinates $(v_1,v_2)^T$. We have non trivial vector spaces $ \wedge^k V $ only for $ k=0,1,2:$
\begin{itemize}
\item
$ \wedge^0 V $ is the space of scalars, it has dimension $1$ and $\{1\}$ is an orthonormal basis;
\item
$ \wedge^1 V = V$ is the space of vectors, it has dimension $2$ and $B$ is an orthonormal basis;
\item
$ \wedge^2 V $ is the space of the $2$--forms, it has dimension $ 1 $ and $\{ \bs{\omega} = \mb{e_1} \wedge \mb{e_2}\}$ is an orthonormal basis.
\end{itemize}
The three spaces $\wedge^0 V,\wedge^1 V,\wedge^2 V$ form the exterior algebra $\wedge V$ over $V.$ The symbol $ \wedge $ stays for the exterior product, which is skew--commutative and linear with respect to each factor. We can be very explicit by working in coordinates with respect to the above natural basis. Let $\mb{v} = (v_1,v_2)^T,\ \mb{w} =(w_1,w_2)^T$ be vectors. Then
$$
\mb{v} \wedge \mb{w} = (v_1w_2-v_2w_1)\,\bs{\omega}=
\det
\left( \begin{array}{cc}
v_1 & w_1 \\
v_2 & w_2
\end{array} \right)
\bs{\omega}\in\wedge^2V.
$$

The Hodge operator $\ast$ defines an isomorphism between each pair of vector spaces $ \wedge^k V $ and $ \wedge^{2-k} V,\ k=0,1,2.$ Also in this case, we can give an explicit definition of $\ast$ by describing its action on the natural basis:
$$
\ast\,1 = \boldsymbol{\omega},\qquad
\ast\,\mb{e_1} = \mb{e_2},\ \ast\,\mb{e_2}=-\mb{e_1},\qquad
\ast\,\boldsymbol{\omega}= 1.
$$
The linearity of $\ast,$ allows us to write $\ast\,\mb{v}=(-v_2,v_1)^T$, which means that the Hodge operator acting on $\wedge^1 V$ corresponds to a counterclockwise rotation of $\frac{\pi}{2}$, represented by the matrix
$$
\mathbf{H} = 
\left( \begin{array}{cc}
0 & -1\\
1 & 0
\end{array} \right)
$$
with respect to $B$, therefore we have $\ast\mb{v}=\mb{H\,v}$.
Finally, we have
$$
\ast(\mb{v} \wedge \mb{w}) = 
\ast((v_1w_2-v_2w_1)\,\bs{\omega})=
\det
\left( \begin{array}{cc}
v_1 & w_1 \\
v_2 & w_2
\end{array} \right).
$$

\subsection{The complete TDOA map}
We identify the physical world with the Euclidean plane and, after choosing an orthogonal Cartesian coordinate system, with $\RR^2$. We use $B$ as the orthonormal basis. On this plane, we have three receivers $\m{i}=(x_{i},y_{i})^T,\ i=0,1,2$ at known positions and a source $\x = (x,y)^T$. The corresponding displacement vectors are
\begin{equation}
\mb{d_i}(\x)=\x-\m{i},\qquad
\mb{d_{ji}}=\m{j}-\m{i},\qquad i,j=0,1,2,
\end{equation}
whose norms are $d_i(\x)$ and $ d_{ji}$, respectively.
Generally speaking, given a vector $\mb{v}$, we denote its Euclidean norm $||\mb{v}||$ with $v$ and with $\tilde{\mb{v}}= \frac{\mb{v}}{v}$ the corresponding unit vector. Furthermore, we name the angles $\alpha=\widehat{\mb{d_{10}}\mb{d_{20}}},\,\beta=\widehat{\mb{d_{01}}\mb{d_{21}}}$ and $\gamma=\widehat{\mb{d_{02}}\mb{d_{12}}}.$

In Figure \ref{rette-acustica} we draw a configuration of the receivers. We set $r_0,r_1,r_2$ the lines containing the sensors, according to the convention that the receiver $\m{i}$ does not lie on $r_i$. Up to relabeling the sensors, we can assume that $\ast(\mb{d_{10}}\wedge\mb{d_{20}})>0,$ i.e. $\mb{d_{10}},\mb{d_{20}}$ are counterclockwise oriented.
\begin{figure}[htb]
\begin{center}
\resizebox{5.7cm}{!}{
  \includegraphics[bb=206 316 405 475]{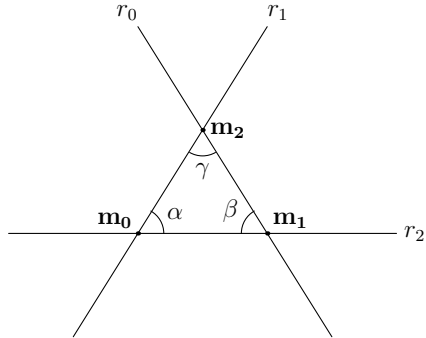}}
  \caption{\label{rette-acustica}Receivers $\m{0},\m{1},\m{2}$ in a generic planar configuration.}
\end{center}
\end{figure}

With no loss of generality, let us set the speed of propagation of the signal in the medium to $1$.
Therefore, in the noiseless scenario, the TDOA between each pair of different sensors is equal to the difference of the ranges:
\begin{equation}\label{TDOA}
\tau_{ji}(\x)=d_{j}(\x)-d_{i}(\x),\quad
i,j=0,1,2.
\end{equation}
We collect the three range differences in the complete TDOA map:
\begin{equation}\label{eq:TDOAmap}\
\begin{array}{cccc}
\bs{\tau_2^*}: & \RR^2 & \longrightarrow & \RR^3\\
& \x & \longmapsto & (\tau_{10}(\x),\tau_{20}(\x),\tau_{21}(\x))^T
\end{array}\ .
\end{equation}

The resulting target set $\RR^3$ of $\bs{\tau_2^*}$ is referred to as the TDOA space or $\tau$--space.
The map $\bs{\tau_2^*}$ completely defines the deterministic model behind the TDOA based source localization. In particular, its image $\text{Im}(\bs{\tau_2^*})$ is the set of feasible TDOAs in the $\tau$--space. This means that three noiseless TDOAs define a point $\bs{\tau^*}= (\tau_{10},\tau_{20},\tau_{21})^T\in\text{Im}(\bs{\tau_2^*})$ and, wherever the map $\bs{\tau_2^*}$ is invertible, the source position is ${\bs{\tau_2^*}}^{-1}(\bs{\tau^*}).$

\subsection{The reduced TDOA map}
The three range differences \eqref{TDOA} are not independent. In fact, the linear relation $ \tau_{21}(\x ) = \tau_{20}(\x ) - \tau_{10}(\x ) $ holds for each $ \x \in \RR^2.$ This means that three noiseless TDOAs are constrained on the plane
\begin{equation}
\mathcal{H}=\{\bs{\tau^*}\in\RR^3\ |\ \tau_{10}-\tau_{20}+\tau_{21}=0\}.
\end{equation}
Therefore, we are allowed to choose $\m{0}$ as a reference sensor and, without loss of information, to consider only the two TDOAs $\tau_{10}(\x),\tau_{20}(\x).$ We define the (reduced) TDOA map:
\begin{equation}
\begin{array}{cccc}
\bs{\tau_2}: & \RR^2          & \longrightarrow & \RR^2\\
 & \x   & \longrightarrow & \quad
(\tau_{10}(\x),\tau_{20}(\x))^T
\end{array}.
\end{equation}

Let us consider the projection map $p_3:\RR^3\rightarrow\RR^2$ forgetting the third coordinate $\tau_{21}$ of the $\tau$--space. Then, we have $\bs{\tau_2}=p_3\circ\bs{\tau_2^*}$ and $p_3$ is a natural bijection between $\text{Im}(\bs{\tau_2^*})$ and $\text{Im}(\bs{\tau_2}).$ Hence, we can investigate the properties of the deterministic TDOA model by studying the simpler map $\bs{\tau_2}$. In analogy with our previous notations, we name $\tau$--plane the target set $\RR^2$ of $\bs{\tau_2}$. To illustrate our exposition, in Figure \ref{fig:tauimage} we draw $\text{Im}(\bs{\tau_2}),$ with receivers $\m{0}=(0,0)^T,\ \m{1}=(2,0)^T$ and $\m{2}=(2,2)^T,$ while in Figure \ref{fig:taucomimage} we show its relation with $\text{Im}(\bs{\tau_2^*}).$ We use this configuration of the sensors in all figures of the manuscript.
\begin{figure}[htb]
\begin{center}
\resizebox{6.5cm}{!}{
  \includegraphics{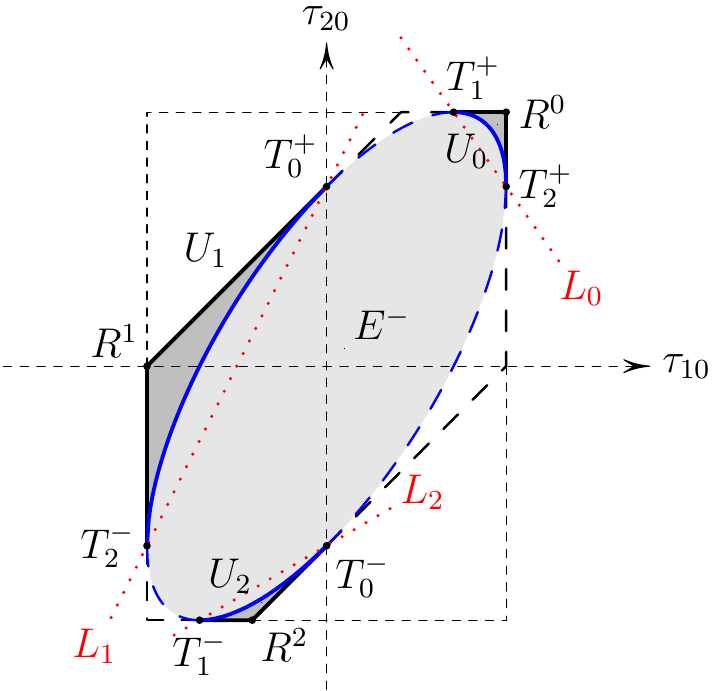}}
  \caption{\label{fig:tauimage}
The image of $\bs{\tau_2}$ is the gray subset of the hexagon $P_2$ with continuous and dashed sides. In the light gray region $E^-$ the map $\bs{\tau_2}$ is $ 1$--to--$1,$ while in the medium gray region $U_0 \cup U_1 \cup U_2$ the map $\bs{\tau_2}$ is $2$--to--$1.$ The continuous part of the boundary of the hexagon and the blue ellipse $E,$ together with the vertices $R^i,$ are in the image, and there $\bs{\tau_2}$ is $ 1$--to--$1.$ The points $T_i^\pm$ and the dashed boundaries do not belong to Im($\bs{\tau_2}$). Finally, the red dotted lines $L_0,L_1,L_2$ allow us to single out the medium gray regions $U_0,U_1,U_2.$}
\end{center}
\end{figure}

\begin{figure}[htb]
\begin{center}
\resizebox{8cm}{!}{
\includegraphics{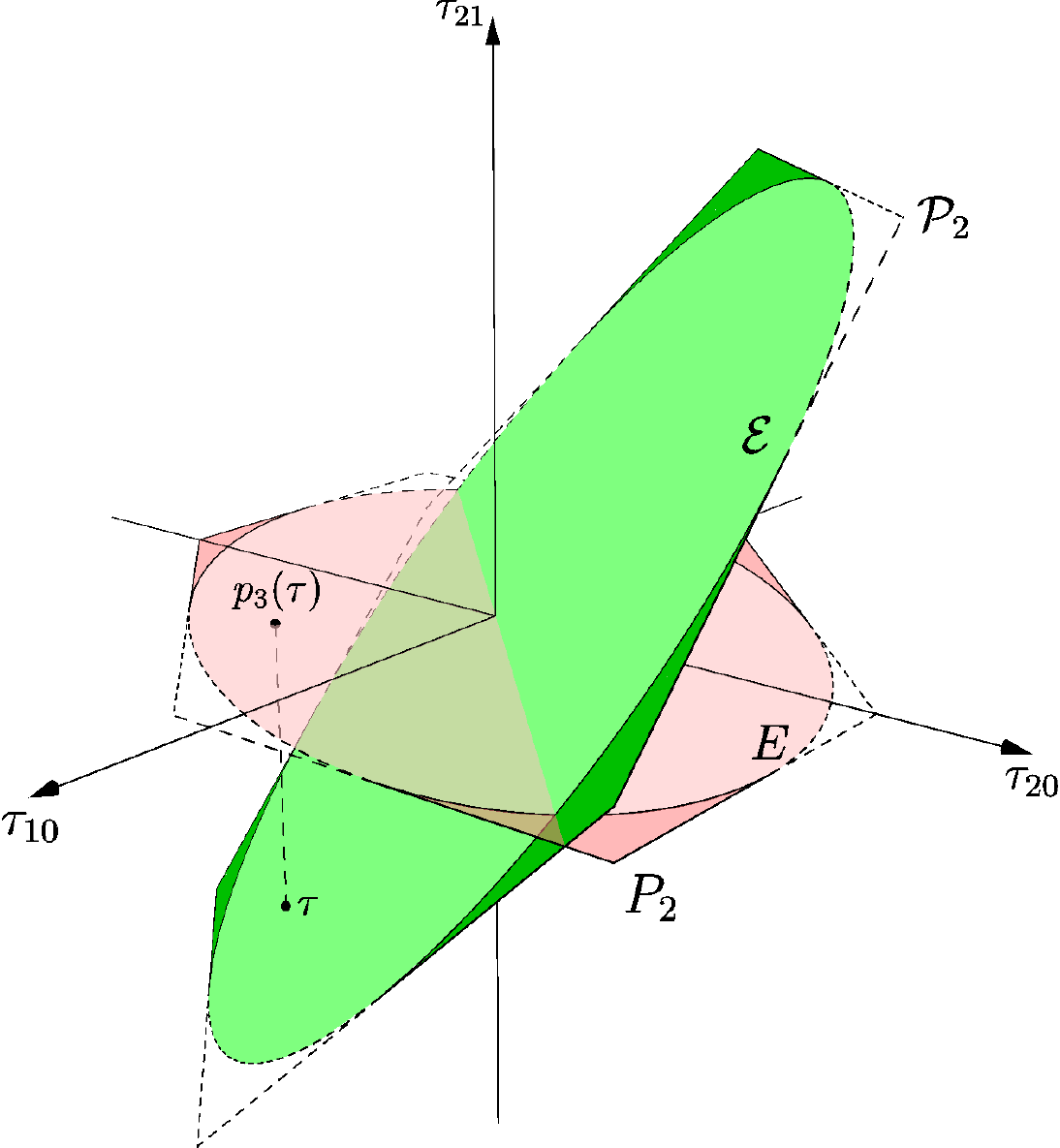}}
\caption{\label{fig:taucomimage}The image of $\bs{\tau_2^*}$ is the green subset of the hexagon $\mathcal{P}_2\subset\mathcal{H},$ while the image of $ \bs{\tau_2}$ is the red subset of $P_2.$ There is a 1--to--1 correspondence between Im($\bs{\tau_2^*}$) and Im($\bs{\tau_2}$) via the projection map $p_3$.   In the lightly shaded regions, the TDOA maps are $1$--to--$1$, while in the more darkly shaded regions the maps are $2$--to--$1$.}
\end{center}
\end{figure}

Following the analysis contained in Section 6 of \cite{Compagnoni2013a}, for any $\bt=(\tau_{10},\tau_{20})^T\in\RR^2$ we define the vectors
\begin{equation}
\begin{array}{l}
\mb{v}(\bt)=\ast(\tau_{20} \mb{d_{10}} - \tau_{10} \mb{d_{20}})\,,\qquad\qquad 
\mb{l_0}(\bt)=\displaystyle
\frac{\ast((d_{20}^2-\tau_{20}^2)\mb{d_{10}} -
(d_{10}^2-\tau_{10}^2)\mb{d_{20}})}
{2 \ast(\mb{d_{10}} \wedge \mb{d_{20}})}
\end{array}
\end{equation}
and the polynomials
\begin{equation}
a(\bt) = \Vert\mb{v}(\bt)\Vert^2 -\ast(\mb{d_{10}}\wedge\mb{d_{20}})^2,\qquad
b(\bt) = \langle\mb{v}(\bt),\mb{l_0}(\bt)\rangle,\qquad
c(\bt) = \Vert\mb{l_0}(\bt)\Vert^2.
\end{equation}
Im$(\bs{\tau_2})$ is a subset of the convex polytope $P_2$, the hexagon defined by the triangle inequalities:
\begin{equation}\label{eq:politopo}
\left\{ \begin{array}{l}
-d_{10} \leq \tau_{10} \leq d_{10} \\
-d_{20} \leq \tau_{20} \leq d_{20} \\
-d_{21} \leq \tau_{20} - \tau_{10} \leq d_{21}
\end{array} \right. .
\end{equation}
The vertices of $P_2$ in $\text{Im}(\bs{\tau_2})$ are $R^0=(d_{10},d_{20})^T,R^1=(-d_{10},d_{21}-d_{10})^T$ and $R^2=(d_{21}-d_{20},-d_{20})^T,$ which are the images of $\m{0},\m{1},\m{2}.$

There exists a unique ellipse $E$ that is tangent to each facet of $P_2$. This ellipse is the one defined by $a(\bt)=0$. We name $E^-$ the interior region of the ellipse, where $a(\bt)<0,$ and $E^+$ the exterior region, where $a(\bt)>0.$ The six points in $E\cap \partial P_2$ are
$$
T_i^+ = \left( \langle \mb{d_{10}},\tilde{\mb{d}}_{\mb{jk}} \rangle, 
\langle \mb{d_{20}}, \tilde{\mb{d}}_{\mb{jk}} \rangle \right)^T
\quad 
\mbox{ and }
\quad 
T_i^- = \left( -\langle \mb{d_{10}}, \tilde{\mb{d}}_{\mb{jk}} 
\rangle, -\langle \mb{d_{20}},\tilde{\mb{d}}_{\mb{jk}} \rangle \right)^T,
$$
where $0 \leq i,j,k \leq 2,\ k<j$ and $j,k \neq i.$
Let us consider the lines $L_0,L_1,L_2$ passing through the couples of points $\{T_1^+,T_2^+\},\{T_0^+,T_2^-\}$ and $\{T_0^-,T_1^-\},$ respectively (see Figure \ref{fig:tauimage}). With straightforward computations we have:
\begin{equation}\label{eq:lines}
\begin{array}{cl}
L_0: &\ l_0(\bt)=d_{20}\tau_{10}+d_{10}\tau_{20}-d_{10}d_{20}(1+\cos\alpha)=0,\\
L_1: &\ l_1(\bt)=-(d_{10}+d_{21})\tau_{10}+d_{10}\tau_{20}-d_{10}d_{21}(1+\cos\beta)=0,\\
L_2: &\ l_2(\bt)=d_{20}\tau_{10}-(d_{20}+d_{21})\tau_{20}-d_{20}d_{21}(1+\cos\gamma)=0.
\end{array}
\end{equation}
Then, we define the three sets $U_i,\ i=0,1,2,$ as:
\begin{equation}\label{eq:Ui}
\begin{array}{l}
U_i=\{\bt\in\mathring{P_2}\,|\,a(\bt)>0,\,l_i(\bt)>0\},\\
\end{array}
\end{equation}
where $\mathring{P_2}$ is the interior of $P_2$ defined by taking the strict inequalities in \eqref{eq:politopo}.

Using the above notation, the image of $\bs{\tau_2}$ is
\begin{equation}
\mbox{Im}(\bs{\tau_2}) = E^- \cup \bar{U}_0 \cup \bar{U}_1 \cup
\bar{U}_2 \setminus \{ T_0^\pm, T_1^\pm, T_2^\pm \},
\end{equation}
where $\bar{U}_i$ stays for the closure of $U_i$ with respect to the Euclidean topology. In particular, we have
\begin{equation}
\vert\bs{\tau_2}^{-1}(\bt)\vert=
\begin{cases}
2 & \text{if }\ \bt \in U_0 \cup U_1 \cup U_2,\\
1 & \text{if }\ \bt \in \mbox{Im}(\bs{\tau_2}) \setminus U_0 \cup U_1 \cup U_2.
\end{cases}
\end{equation}
For any given $ \bt \in U_0 \cup U_1 \cup U_2,$ the two preimages $\x_\pm(\bt)$ are given by
\begin{equation}\label{eq:inv-image}
\x_\pm(\bt) = \mb{m_0}+\mb{l_0}(\bt) + \lambda_\pm(\bt) \mb{v}(\bt),
\end{equation}
where $ \lambda_\pm(\bt) $ are the solutions of the quadratic equation $a(\bt)\lambda^2+2b(\bt)\lambda+c(\bt)=0:$
\begin{equation}\label{eq:lambda}
\lambda_{\pm}(\bt)=\frac{-b(\bt)\pm\sqrt{b(\bt)^2-a(\bt)c(\bt)}}{a(\bt)}\,.
\end{equation}
For $ \bt \in \mbox{Im}(\bt_2) \setminus U_0 \cup U_1 \cup U_2,$ we have to take only the $\x_+(\bt)$ solution.

In Figure \ref{fig:x-plane} we give two examples of the different localization regions in the $x$--plane. Roughly speaking, we have the preimage of the interior of the ellipse $ \tilde{E}^- =\bs{\tau_2}^{-1}(E^-)$, where the TDOA map is $1$--to--$1$ and the source localization is possible, and the preimages $ \tilde{U}_i=\bs{\tau_2}^{-1}(U_i)$, for $ i=0,1,2,$ where the map is $2$--to--$1$ and there is no way to uniquely locate the source. The transition is on the \emph{bifurcation curve} $\tilde{E} =\bs{\tau_2}^{-1}(E),$ that consists of three disjoint and unbounded arcs, one for each arc of $ E $ contained in $ \mbox{Im}(\bs{\tau_2}).$ As a point $\bt$ in one of the $U_i$ gets close to $E$, the solution $\x_+(\bt)$ gets close to a point on $\tilde{E}$, while $\x_-(\bt)$ goes to infinity. The sets $\tilde{E}^-, \tilde{U}_0, \tilde{U}_1, \tilde{U}_2 $ are open subsets of the $ x$--plane, separated by the three arcs of $ \tilde{E} $.

\begin{figure}[htb]
\begin{center}
\resizebox{11cm}{!}{
  \includegraphics[bb=-0 -0 1654 667]{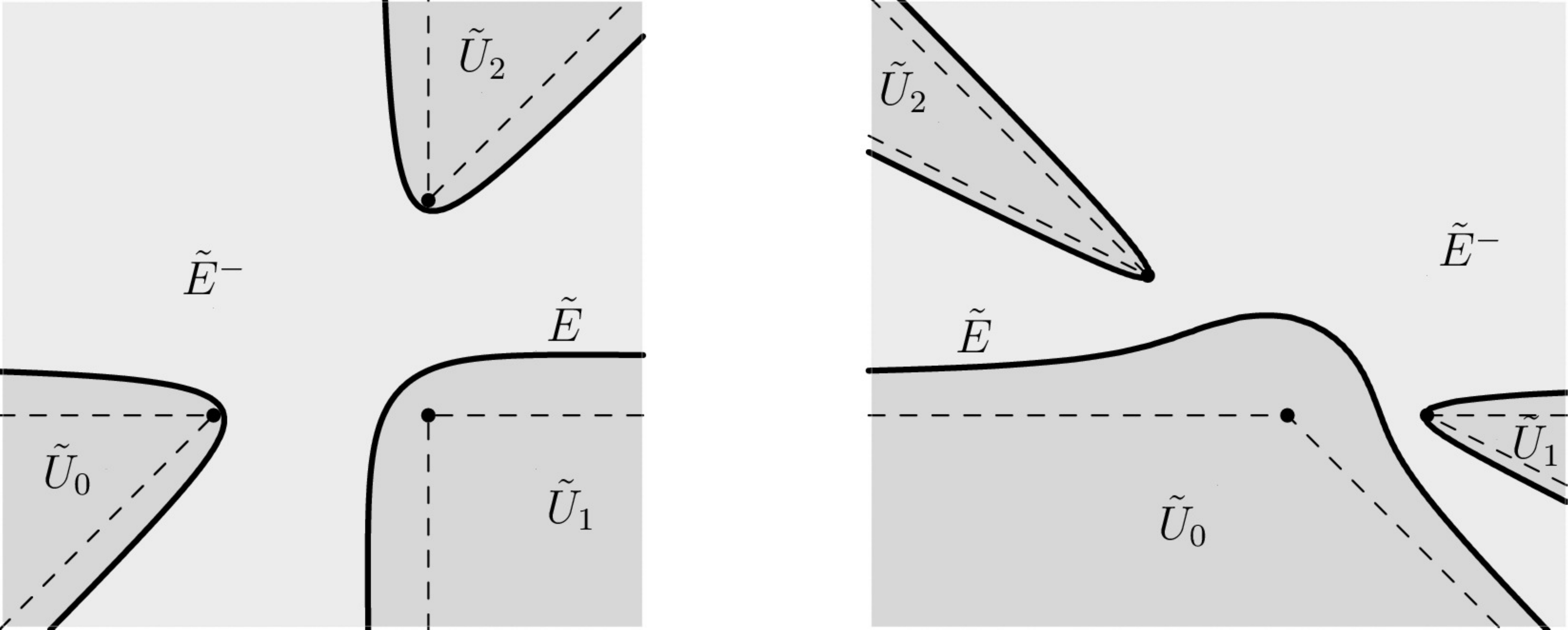}}
\caption{\label{fig:x-plane}Two examples of the different localization regions and the curve $ \tilde{E} $ in the $x$--plane. The sensors are the marked points $ \m{0}=(0,0)^T,\ \m{1}=(2,0)^T,$ and either $ \m{2}=(2,2)^T $ on the left, or $ \m{2}=(-2,2)^T $ on the right. Each curve $\tilde{E}$ separates the light gray region $\tilde{E}^-$, where the map $\bt_2$ is 1--1 and it is possible to locate the source, and the medium gray region $\tilde{U}_0\cup\tilde{U}_1\cup\tilde{U}_2$, where $\bs{\tau_2}$ is 2--1 and the localization is not unique. On the dashed lines the localization is possible but very sensitive to the measurement noise.}
\end{center}
\end{figure}

Finally, the union $D$ of the six dashed half--lines originating from the receivers is called \emph{degeneracy locus} of the TDOA map, where the rank of the Jacobian matrix of $\bs{\tau_2}$ drops. $D$ is the zero set of the Jacobian
\begin{equation}\label{eq:degloc}
D:\;\ast(\tilde{\mb{d}}_{\mb{1}}(\x) \wedge \tilde{\mb{d}}_{\mb{0}}(\x)-
\tilde{\mb{d}}_{\mb{2}}(\x) \wedge \tilde{\mb{d}}_{\mb{0}}(\x) +
\tilde{\mb{d}}_{\mb{2}}(\x) \wedge\tilde{\mb{d}}_{\mb{1}}(\x))=0
\end{equation}
and it is the preimage of the six segments in $\partial P_2\cap\text{Im}(\bs{\tau_2}).$
On $D$ the two solution $\x_\pm(\bt)$ are coincident, thus the TDOA map is $1$--to--$1.$ Furthermore, $D$ divides each $ \tilde{U}_i $ into two connected components and $\bs{\tau_2}$ is a bijection between each of them and the corresponding $U_i$. For future reference, we observe that the lines $r_i$ supporting $D$ (see Figure \ref{rette-acustica}) have equations
\begin{equation}\label{eq:rette}
r_0: \; \ast(\mb{d_2}(\x) \wedge \mb{d_1}(\x))=0,\qquad
r_1: \; \ast(\mb{d_2}(\x) \wedge \mb{d_0}(\x))=0,\qquad
r_2: \; \ast(\mb{d_1}(\x) \wedge \mb{d_0}(\x))=0.
\end{equation}


\section{The statistical model}\label{sc:StatMod}
In the presence of measurement errors on the data, we must resort to statistical modeling. In this section, we focus on the definition of the statistical models for TDOA--based localization in the minimal planar scenario. As we will see, we will need to define a plurality of models in order to take care of issues of non-uniqueness in source localization. In particular, we will consider four distinct curved exponential families, corresponding to the four different regions in the $x$--plane.
In our analysis we will follow the notation of \cite{Amari2000}. This will be particularly useful in Section \ref{sc:ASI}, where we study the source estimation accuracy via the asymptotic analysis techniques given by Information Geometry.

\subsection{The complete and the reduced models}
In this subsection we adapt the analysis contained in Sections 3 and 4 of \cite{Compagnoni2016a} to the case of three TDOAs. For sake of completeness, we include in this Section part of the mathematical derivation in \cite{Compagnoni2016a}. In this manuscript we assume the noise to be Gaussian \cite{Benesty2004}, therefore the TDOAs associated to a source in $\x$ are described by
\begin{equation}\label{eq:TDOAstatmod}
\bs{\hat\tau_2^*}(\x)=\bs{\tau_2^*}(\x)+\bs{\epsilon},\qquad
\text{where}\qquad \bs{\epsilon}\sim N(\mb{0},\bs\Sigma)
\end{equation}
and the covariance matrix $\bs\Sigma$ is known and non singular. This is the most common choice in application scenarios \cite{Spencer2010,Kaune2011}. From a mathematical standpoint, this error distribution allows us to use the many existing tools for the study of statistical exponential families. However, our analysis is helpful also in situations where errors are no longer Gaussian, e.g. in presence of outliers due to phenomena such as interferer sources or multipath propagation of the signal. For example, in \cite{Compagnoni2016c} it has been defined an outliers removal procedure that identifies a TDOA as an outlier exactly when it does not satisfy the gaussianity assumption.

We define \eqref{eq:TDOAstatmod} as the \emph{complete statistical model}.
This means that the probability density function (p.d.f.) for the measured TDOAs $\bs{\hat\tau^*}=(\hat\tau_{10},\hat\tau_{20},\hat\tau_{21})^T$ is
\begin{equation}\label{eq:TDOAprob}
p(\bs{\hat\tau^*};\bs{\tau_2^*}(\x),\bs\Sigma)=
\frac{1}{\sqrt{(2\pi)^3|\bs\Sigma|}}\,
\exp\left[-\frac{1}{2}(\bs{\hat\tau^*}-\bs{\tau_2^*}(\x))^T\; \bs\Sigma^{-1}(\bs{\hat\tau^*}-\bs{\tau_2^*}(\x))\right]\,.
\end{equation}

From a geometric standpoint, the Fisher matrix $\bs\Sigma^{-1}$ defines a Euclidean structure on the $\tau$--space, with  scalar product
\begin{equation}\label{eq:Fisherscalarprod}
\langle\mb{v_1},\mb{v_2}\rangle_{\boldsymbol\Sigma^{\unaryminus 1}} = \mb{v_1}^T\; \boldsymbol\Sigma^{-1} \mb{v_2},\qquad
\mb{v_1},\mb{v_2}\in\RR^3 \; .
\end{equation}
The associated distance is the Mahalanobis distance
\begin{equation}\label{eq:Mahalanobis}
\Vert\mb{v}\Vert_{\bs\Sigma^{\unaryminus 1}}=
\sqrt{\langle\mb{v},\mb{v}\rangle_{\boldsymbol\Sigma^{\unaryminus 1}}}\,,
\qquad \mb{v}\in\RR^3
\end{equation}
and we can rewrite the p.d.f. \eqref{eq:TDOAprob} as:
\begin{equation}\label{eq:TDOAprob2}
p(\bs{\hat\tau^*};\bs{\tau_2^*}(\x),\bs\Sigma)=
\frac{1}{\sqrt{(2\pi)^3|\bs\Sigma|}}\,
\exp\left[-\frac{1}{2}\,
\Vert\bs{\hat\tau^*}-\bs{\tau_2^*}(\x))\Vert_{\bs\Sigma^{\unaryminus 1}}^2\right]\,.
\end{equation}
\begin{theorem}\label{th:ss}
Let $\mathcal{P_H}(\bs{\hat\tau^*};\bs\Sigma)$ be the orthogonal projection of $\bs{\hat\tau^*}\in\RR^3$ on the plane $\mathcal{H}$ defined by $\tau_{10}-\tau_{20}-\tau_{21}=0$, with respect to $\langle\ ,\ \rangle_{\Sigma^{\unaryminus 1}}.$ Then, $\bs{\hat\tau}=p_3\circ\mathcal{P_H}\,(\bs{\hat\tau^*};\bs\Sigma)$ is a sufficient statistic for the underlying parameter $\x.$
\end{theorem}
\begin{proof}
$\mathcal{P_H}(\bs{\hat\tau^*};\bs\Sigma)$ is a sufficient statistic for $\x,$ see Theorem 1 in \cite{Compagnoni2016a}. Since the forgetting map $p_3$ is $1$--$1$ between $\mathcal{H}$ and $\RR^2,$ the claim follows.
\end{proof}
Theorem \ref{th:ss} states that all the information about the source position is contained in $\bs{\hat\tau}\in\RR^2$. In order to obtain the p.d.f. for $\bs{\hat\tau},$ we observe that $p_3\circ\mathcal{P_H}\,(\bs{\tau_2^*}(\x);\bs\Sigma)=\bs{\tau_2}(\x)$ and we define $\bs\Sigma_2=\mb{P}\bs\Sigma\mb{P}^T,$ where $\mb{P}$ is the representative matrix of $p_3\circ\mathcal{P_H}$ with respect to the standard basis of $\RR^3$ and $\RR^2$ (see \cite{Compagnoni2016a}). From the general transformation rule for the multivariate normal distributions under linear mapping, it follows:
\begin{equation}\label{eq:TDOAprobproj}
p(\bs{\hat\tau};\bs{\tau_2^*}(\x),\bs\Sigma)=p(\bs{\hat\tau};\bs{\tau_2}(\x),\bs{\Sigma_2})=
\frac{1}{2\pi\sqrt{|\bs{\Sigma_2}|}}\,
\exp\left[-\frac{1}{2}\,
\Vert\bs{\hat\tau}-\bs{\tau_2}(\x))\Vert_{\bs{\Sigma_2}^{\unaryminus 1}}^2\right]\,.
\end{equation}
where it appears the Mahalanobis distance defined by $\bs{\Sigma_2}^{-1}$ on $\RR^2$. This means that the analysis of the complete statistical model \eqref{eq:TDOAstatmod} is equivalent to the analysis of the reduced (2D) TDOA statistical model:
\begin{equation}\label{eq:TDOAstatmodrid}
\bs{\hat\tau_2}(\x)=\bs{\tau_2}(\x)+\bs{\epsilon_2},\qquad \text{where}\qquad \bs{\epsilon_2}\sim N(\mb{0},\bs{\Sigma_2}).
\end{equation}
For this reason and without loss of generality, in the rest of the paper we focus on the analysis of \eqref{eq:TDOAstatmodrid}.

\subsection{The restricted models are curved exponential families}
We now address the problem of ambiguity in source localization, i.e. the fact that the map $\bs{\tau_2}$ is not globally invertible. The simplest way to solve this issue is to define distinct statistical models for each of the maximal subsets of the $x$--plane where $\bs{\tau_2}$ is injective. Throughout the rest of the manuscript we call these the \emph{restricted models}. In the following proposition we define and study the properties of the above maximal subsets of $\RR^2$.
\begin{proposition}\label{prop:restmod}
In the $x$--plane, let us define the subsets
\begin{equation*}
\begin{array}{l}
\Omega=\{\x\in\RR^2|\ast(\tilde{\mb{d}}_{\mb{1}}(\x) \wedge \tilde{\mb{d}}_{\mb{0}}(\x)-
\tilde{\mb{d}}_{\mb{2}}(\x) \wedge \tilde{\mb{d}}_{\mb{0}}(\x) +
\tilde{\mb{d}}_{\mb{2}}(\x) \wedge\tilde{\mb{d}}_{\mb{1}}(\x))<0\},\\[.5mm]
\Omega_0=\{\x\in\RR^2|\ast(\mb{d_1}(\x) \wedge \mb{d_0}(\x))>0,
\ast(\mb{d_2}(\x) \wedge \mb{d_0}(\x))<0\},\\[.5mm]
\Omega_1=\{\x\in\RR^2|\ast(\mb{d_1}(\x) \wedge \mb{d_0}(\x))>0,
\ast(\mb{d_2}(\x) \wedge \mb{d_1}(\x))>0\},\\[.5mm]
\Omega_2=\{\x\in\RR^2|\ast(\mb{d_2}(\x) \wedge \mb{d_0}(\x))<0,
\ast(\mb{d_2}(\x) \wedge \mb{d_1}(\x))>0\}.
\end{array}
\end{equation*}
Then:
\begin{enumerate}
\item
the above subsets are open and disjoint from each other and their union is dense in $\RR^2$. In particular, $\RR^2\setminus(\Omega\cup\Omega_0\cup\Omega_1\cup\Omega_2)=D,$ the degeneracy locus of $\bs{\tau_2};$
\item
the restriction of $\bs{\tau_2}$ on each subset is differentiable and injective, with $\text{Im}(\bs{\tau_2}|_\Omega)=\text{Im}(\bs{\tau_2})\setminus\partial P_2$ and $\text{Im}(\bs{\tau_2}|_{\Omega_i})=U_i,\ i=0,1,2.$ Furthermore, we have $\left(\bs{\tau_2}|_\Omega\right)^{-1}=\x_+$ and $\left(\bs{\tau_2}|_{\Omega_i}\right)^{-1}=\x_-,\ i=0,1,2.$
\end{enumerate}
\end{proposition}
\begin{proof}
(1) For each $i=0,1,2,$ the subset $\Omega_i$ is defined as the intersection of two open half--planes. From equations \eqref{eq:rette} defining $r_0,r_1,r_2$, it is straightforward to verify that $\Omega_i$ corresponds to the open subsets with vertex $\m{i}$ drawn in Figure \ref{fig:spazio-x}. This implies that $\Omega_i\cap\Omega_j=\emptyset$ if $i\neq j$ and the boundary of $\displaystyle\cup_{i=0}^2\Omega_i$ is $D.$ On the other hand, from equation \eqref{eq:degloc} we have that $D$ is the boundary of $\Omega.$ Moreover, if $\x\in\Omega$ one has $\ast(\mb{d_1}(\x) \wedge \mb{d_0}(\x))<0,$ $\ast(\mb{d_2}(\x) \wedge \mb{d_0}(\x))>0$ and $\ast(\mb{d_2}(\x) \wedge \mb{d_1}(\x))<0.$ Therefore
$\ast(\tilde{\mb{d}}_{\mb{1}}(\x) \wedge \tilde{\mb{d}}_{\mb{0}}(\x)-
\tilde{\mb{d}}_{\mb{2}}(\x) \wedge \tilde{\mb{d}}_{\mb{0}}(\x) +
\tilde{\mb{d}}_{\mb{2}}(\x) \wedge\tilde{\mb{d}}_{\mb{1}}(\x))<0$ for every $\x\in\Omega.$ Since the Jacobian of $\bt_2$ changes its sign on $D,$ this proves that $\Omega$ is the remaining open subset in Figure \ref{fig:spazio-x} and the first claim follows.\\
(2) The second claim is a consequence of the first and the properties of $\bt_2$ proved in \cite{Compagnoni2013a} and summarized in Section \ref{sc:TDOAspace}.
\end{proof}

\noindent In order to simplify the notations, let us define $U=\text{Im}(\bt_2|_\Omega).$ An explicit description of $U$ follows easily by considering in the $\tau$--plane the cubic curve $C$ defined by equation $b(\bt)=0$ and the associated open regions $C^+$ and $C^-,$ defined as $b(\bt)>0$ and $b(\bt)<0$ respectively (see Figure \ref{fig:Ellisse-Politopo}). From the results in Section $6$ of \cite{Compagnoni2013a}, we have:
\begin{equation}\label{eq:U}
U=\mathring{P_2}\cap(E^-\cup C^+)=
\{\bt\in\mathring{P_2}\,|\,a(\bt)<0\text{ or }b(\bt)>0\}.
\end{equation}
In Figure \ref{fig:spazio-x} we draw the sets $\Omega,\Omega_0,\Omega_1,\Omega_2,$ while in Figure \ref{fig:modelli} there are their images $U$ and $U_i,\ i=0,1,2.$
\begin{figure}[htb]
\begin{center}
\resizebox{5.5cm}{!}{
  \includegraphics[bb=206 296 405 495]{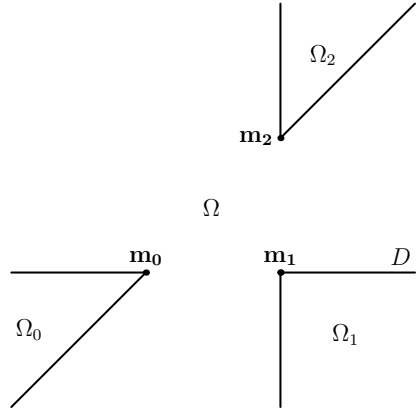}}\\
  \caption{$\Omega,\Omega_0,\Omega_1,\Omega_2$ are disjoint open subsets of the $x$--plane. The closure of their union is $\RR^2,$ while their boundaries give the discriminacy locus $D$ of $\bs{\tau_2}.$ On each subset the map $\bs{\tau_2}$ is injective.}
  \label{fig:spazio-x}
\end{center}
\end{figure}
\begin{figure}[htb]
\begin{center}
\resizebox{6cm}{!}{
  \includegraphics[bb=202 267 409 524]{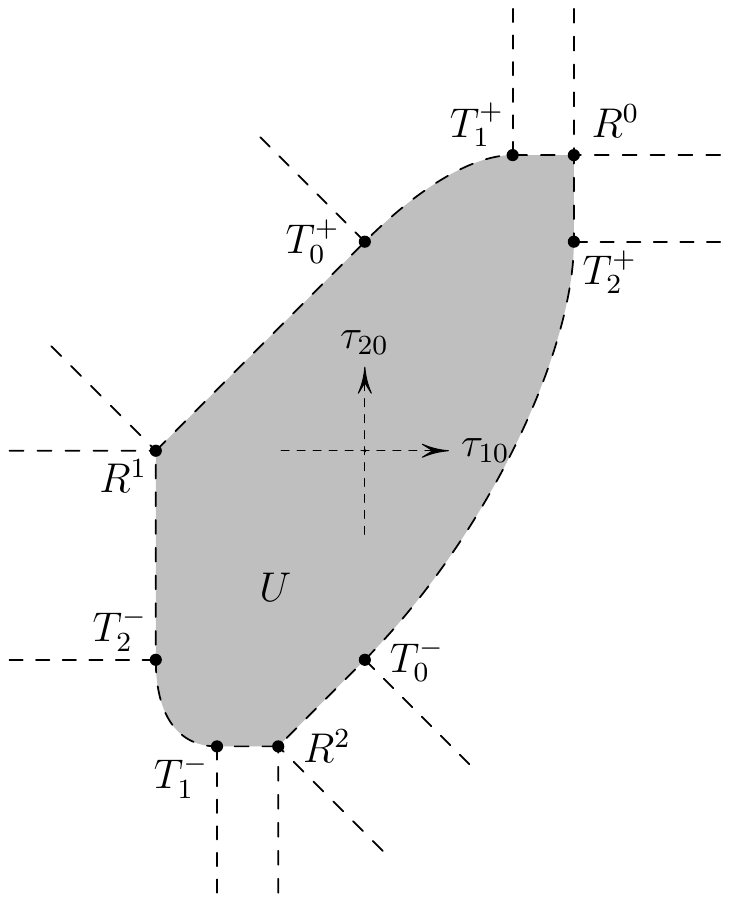}}
  \hspace{15mm}
\resizebox{6cm}{!}{
  \includegraphics[bb=202 267 409 524]{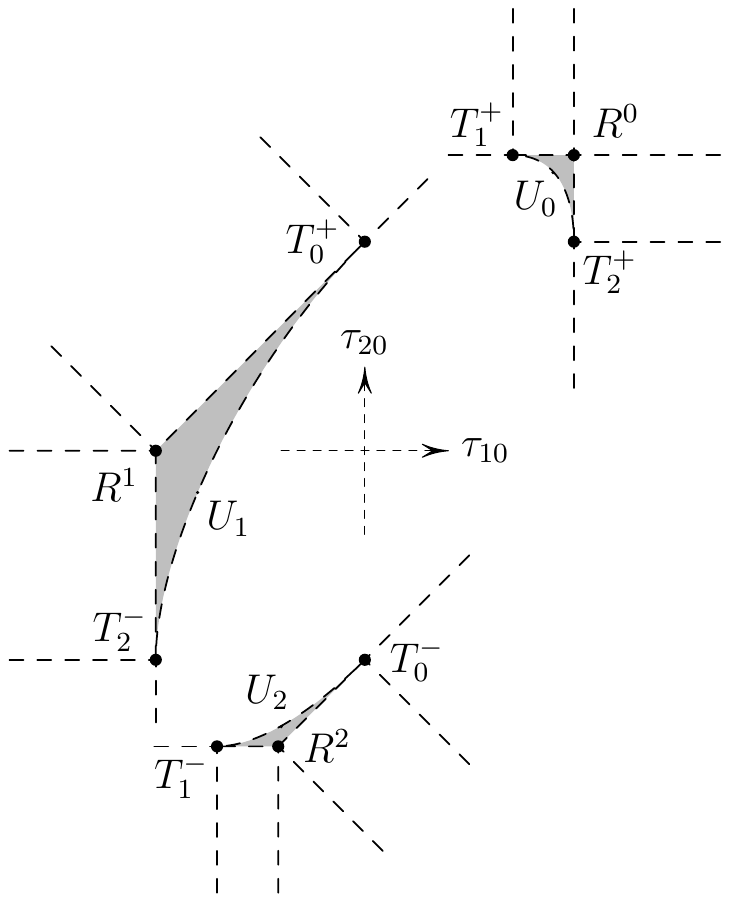}}\\
  \caption{On the left, the medium gray subset is $U=\text{Im}(\bs{\tau_2}|_\Omega)$. On the right, the medium gray subset having $R^i$ as vertex is $U_i=\text{Im}(\bs{\tau_2}|_{\Omega_i}),\ i=0,1,2.$ The dashes lines subdivide the $\tau$--plane into different regions according to the solution of the Maximum Likelihood Estimator for each restricted model (see Section \ref{sc:MLE}).}
  \label{fig:modelli}
\end{center}
\end{figure}

At this point, following \cite{Amari2000}, we recall the definition of curved exponential family.
\begin{definition}\label{def:expfam}
Let $\mathcal{X}$ be a set and $\Theta\subseteq\RR^n$ an open subset. An $n$--dimensional exponential family on $\mathcal{X},$ with parameters $\bs\theta\in\Theta$ and random variables $\mb{y}\in\mathcal{X},$ is a set $S$ of probability density functions
$$
p(\mb{y};\bs{\theta})=
\exp\left[C(\mb{y})+\sum_{i=1}^n\theta_iF_i(\mb{y})-\psi(\bs\theta)\right],
$$
where $\{C,F_1,\ldots,F_n\}$ are $n+1$ real valued functions on $\mathcal{X}$ and $\psi\in C^\infty(\Theta,\RR).$
\end{definition}
\noindent Given an exponential family $S,$ the mapping $\phi:S\rightarrow\Theta$ given by $\phi(p(\mb{y};\bs\theta))=\bs\theta$ is a global coordinate system of $S.$ By composing $\phi$ with every $C^\infty$ diffeomorphism of $\RR^n,$ we obtain a $C^\infty$ atlas on $S.$ This allows us to consider $S$ as a $C^\infty$ differentiable manifold, a so called \emph{statistical manifold}.

The $n$--dimensional multivariate normal distribution is an important example of exponential family. Indeed, let $\mb{y}\sim N(\bs\mu,\bs\Sigma).$ In this context, the probability density function becomes
\begin{equation}\label{eq:MND}
\begin{array}{lcl}
p(\mb{y};\bs\mu,\bs\Sigma) & = &
\displaystyle\frac{1}{\sqrt{(2\pi)^n|\bs\Sigma|}}\,
\exp\left[-\frac{1}{2}\,
\Vert\mb{y}-\bs\mu\Vert_{\bs\Sigma^{\unaryminus 1}}^2\right]\\
& = & \displaystyle\exp\left[
-\frac{1}{2}\,\Vert\mb{y}\Vert_{\bs\Sigma^{\unaryminus 1}}^2
+\langle\bs\mu,\mb{y}\rangle_{\bs\Sigma^{\unaryminus 1}}
-\left(\frac{1}{2}\,\Vert\bs\mu\Vert_{\bs\Sigma^{\unaryminus 1}}^2
+\log\sqrt{(2\pi)^n|\bs\Sigma|}\right)\right]\,.
\end{array}
\end{equation}
If we assume $\bs\Sigma$ to be known, the only parameters are $\bs{\theta}=\bs\mu.$ By defining
$$
C(\mb{y})=-\frac{1}{2}\,\Vert\mb{y}\Vert_{\bs\Sigma^{\unaryminus 1}}^2\,,\qquad
F_i(\mb{y})=(\bs\Sigma^{-1}\mb{y})_i\,,\qquad
\psi(\bs{\theta})=\frac{1}{2}\,\Vert\bs\theta\Vert_{\bs\Sigma^{\unaryminus 1}}^2
+\log\sqrt{(2\pi)^n|\bs\Sigma|}\,,
$$
the p.d.f. \eqref{eq:MND} can be rewritten in the canonical form given in Definition \ref{def:expfam}, where one can identify $\mathcal{X}=\Theta=\RR^n$.
\begin{definition}
An $(n,m)$ curved exponential family on $\mathcal{X}$ is a set $M$ of probability density functions which forms a smooth $m$--dimensional submanifold within an $n$--dimensional exponential family $S.$
\end{definition}

Now, we can state the main result of this Section.
\begin{theorem}\label{th:curvedexpfam}
The restriction of the statistical model $\bs{\hat\tau_2}(\x)$ on each subset $\Omega$ and $\Omega_i,\ i=0,1,2,$ is a $(2,2)$ curved exponential family on $\RR^2,$ parameterized by $\bs\theta(\x)=\bs{\tau_2}(\x).$
\end{theorem}
\begin{proof}
Let us take as $S$ the $2$--dimensional exponential family on $\mathcal{X}=\Theta=\RR^2$ given by \begin{equation}\label{eq:pdfS}
p(\bs{\hat\tau};\bs{\tau_2},\bs{\Sigma_2}) =
\frac{1}{2\pi\sqrt{|\bs{\Sigma_2}|}}\,
\exp\left[-\frac{1}{2}\,
\Vert\bs{\hat\tau}-\bs{\tau_2}\Vert_{\bs{\Sigma_2}^{\unaryminus 1}}^2\right]\,,
\end{equation}
where the parameter space is the $\tau$--plane and $\bs\theta=\bs{\tau_2}.$ By Proposition \eqref{prop:restmod}, on each subset $\Omega$ and $\Omega_i,\ i=0,1,2,$ the map $\bs{\tau_2}(\x)$ is a differentiable bijection and defines a $2$--dimensional family $M$ of distributions $p(\bs{\hat\tau};\bs{\tau_2}(\x),\bs{\Sigma_2})$ which is smoothly embedded in $S.$
\end{proof}
\noindent From now on, we will refer to the curved exponential family defined by $\bs{\hat\tau_2}|_\Omega$ and $\bs{\hat\tau_2}|_{\Omega_i},\ i=0,1,2,$ as $M$ and $M_i,$ respectively, which have the open subsets $U$ and $U_i$ of $\Theta=\RR^2$ as parameter spaces (see Figure \ref{fig:modelli}).
As said at the beginning of the Section, in the rest of the manuscript we will investigate the properties of the TDOA statistical model through the analysis of the families $M,M_0,M_1,M_2.$

\section{The estimation of the source position in the TDOA space}\label{sc:sourceest}
In this Section we consider the problem of the estimation of the source position in the presence of noisy TDOA measurements. Now that we have a precise description of the feasible set of TDOA measurements, we can address the source localization problem in a radically different fashion with respect to the existing literature. The typical approach for estimating the source position, in fact, is based on the optimization of a cost function $f(\x;\bs{\hat{\tau},\Sigma})$ in the $x$--space, and the most well-known example of the sort is the Maximum Likelihood Estimation. In our minimal sensors scenario, if we choose $\bs{\hat{\tau}}\in\RR^2$, the MLE approach consists of maximizing of the likelihood function
$$
l(\bs{\tau_2}(\x);\bs{\hat\tau},\bs{\Sigma_2})=
\frac{1}{2\pi\sqrt{|\bs\Sigma_2|}}\,
\exp\left[-\frac{1}{2}\,
\Vert\bs{\hat\tau}-\bs{\tau_2}(\x))\Vert_{\bs\Sigma_2^{\unaryminus 1}}^2\right]\,,
$$
which returns the location
\begin{equation}\label{eq:MLE}
\mb{\bar{x}}(\bs{\hat\tau};\bs{\Sigma_2})=\underset{\x\in\RR^2}{\text{argmax}}\
l(\bs{\tau_2}(\x);\bs{\hat\tau},\bs{\Sigma_2})=
\underset{\x\in\RR^2}{\text{argmin}}\
\Vert\bs{\hat\tau}-\bs{\tau_2}(\x))\Vert_{\bs{\Sigma_2}^{\unaryminus 1}}^2\;.
\end{equation}
The MLE is an optimal estimator, as its variance asymptotically attains the Cramer--Rao lower bound. However, finding the solution of the optimization problem \eqref{eq:MLE} is a challenging task. The cost function, in fact, is strongly nonlinear, which makes it quite difficult to find a closed-form implementation. We must therefore resort to iterative techniques, which start from a random location and follow a gradient descent. However, as the cost function is not convex, the solution could get easily trapped in a local minimum. In order to reduce the occurrance and the impact of spurious localizations, tracking algorithms are typically used \cite{Antonacci2007}, but for many real-time applications the implementation could result too cumbersome and computationally intensive.

A relevant advantage of our approach to TDOA--based localization via the TDOA--space is that it gives new information and a better control over the optimization problem \eqref{eq:MLE}. In this section we show this fact in our minimal case of three sensors.

\subsection{The Maximum Likelihood Estimation in the TDOA space}\label{sc:MLE}
Let us consider, for the sake of simplicity, the model $M$. Given $\bs{\hat\tau}\in\RR^2,$ the Maximum Likelihood Estimation of the TDOAs is $\bs{\bar\tau}(\bs{\hat\tau};\bs{\Sigma_2})=\bs{\tau_2}(\bs{\bar\x}(\bs{\hat\tau};\bs{\Sigma_2})),$ which satisfies
\begin{equation}\label{eq:MLE-TDOAspace}
\bs{\bar\tau}(\bs{\hat\tau};\bs\Sigma_2) = \underset{\bt\in U}{\text{argmin}}\
\Vert\bs{\hat\tau}-\bt\Vert_{\bs{\Sigma_2}^{\unaryminus 1}}^2\;.
\end{equation}
This means that, in the $\tau$--plane, an MLE algorithm searches for the point $\bs{\bar\tau}(\bs{\hat\tau};\bs{\Sigma_2})\in U$ at minimum Mahalanobis distance from $\bs{\hat\tau}.$
We have two main cases. If $\bs{\hat\tau}\in U$, the MLE solution is simply $\bs{\bar\tau}(\bs{\hat\tau};\bs{\Sigma_2})=\bs{\hat\tau}$. On the other hand, if $\bs{\hat\tau}\notin U$, we have to find the closest point to $\bs{\hat\tau}$ on the boundary $\partial U$ of $U$.  In Figure \ref{fig:modelli}, the dashed lines subdivide the $\tau$--plane in several subsets, according to the different types of solution of the MLE for $\bs{\Sigma_2}=\sigma^2\bs I.$ In $\RR^2\setminus U$ there are:
\begin{itemize}
\item
six regions having the segments of $P_2\cap\partial U$ as boundaries;
\item
three regions having the arcs of $E\cap \partial U$ as boundaries;
\item
three angular regions with vertices $R^i,\ i=0,1,2.$
\end{itemize}
In the first two subcases, the MLE solution $\bs{\bar\tau}(\bs{\hat\tau};\bs{\Sigma_2})$ is the closest orthogonal projection of $\bs{\hat\tau}$ on the relative boundary of $U.$ In the latter, the MLE solution is the corresponding vertex $R^i.$

Similar arguments apply also to the models $M_i,\ i=0,1,2$ and for a generic covariance matrix $\bs{\Sigma_2}.$ In particular, for every model it is necessary to compute the projections of $\bs{\hat\tau}$ on $P_2\cap\bar{U}$ and on the ellipse. The rest of this Section will be focused on the geometric problem of projecting $\bs{\bar\tau}$ on the set of feasible measurements. We leave to Section \ref{sec:validation} the formulation and validation of the corresponding MLE algorithm.

\subsection{The orthogonal projections on the line segments in $P_2\cap\partial U$}\label{sc:projls}
For a generic point $\bs{\hat\tau}\in\RR^2,$ there exists a projection $\mathcal{P}_i^\pm(\bs{\hat\tau};\bs{\Sigma_2})$ on each of the lines $s_i^\pm,\ i=0,1,2,$ supporting the six facets of $P_2$ (see Figure \ref{fig:Ellisse-Politopo}).\footnote{In order to simplify the exposition, herein we adopt a different notation for the facets of $P_2$, with respect to the one used in \cite{Compagnoni2013a}. Indeed, we name $s_i^\pm$ the lines supporting two facets containing the vertex $R^i,$ while in \cite{Compagnoni2013a} we used $F_i^\pm$ for the two parallel facets not containing $R^i.$}
\begin{figure}[htb]
\begin{center}
\resizebox{6.5cm}{!}{
  \includegraphics[bb=191 253 420 538]{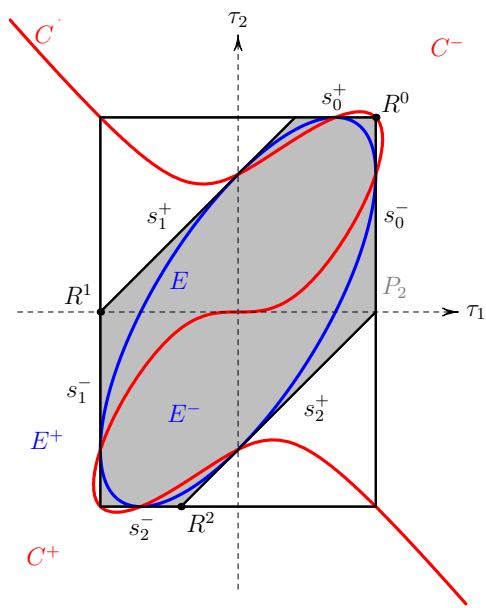}}\\
  \caption{The ellipse $E$ in blue, the cubic $C$ in red and the six lines supporting the facets of $P_2$. Regions $E^\pm$ and $C^\pm$ are named according to the sign of $a(\bt)$ and $b(\bt)$ respectively. We note that $U_0,U_1,U_2\subset C^+.$}
  \label{fig:Ellisse-Politopo}
\end{center}
\end{figure}
\noindent Let us define the vectors $\mb{v_0}=(1,1)^T,\ \mb{v_1}=(1,0)^T\ \text{and}\ \mb{v_2}=(0,1)^T,$ each one parallel to two facets of $P_2.$
Hence, the points $\mathcal{P}_i^\pm(\bs{\hat\tau};\bs{\Sigma_2}),\ i=0,1,2,$ are:
\begin{equation}\label{eq:projfacets}
\begin{array}{lcl}
\mathcal{P}_0^+(\bs{\hat\tau};\bs{\Sigma_2})=R^0+
\displaystyle\frac{\langle\bs{\hat\tau}-R^0,\mb{v_1}\rangle_{\bs{\Sigma_2}^{\unaryminus 1}}}{\Vert\mb{v_1}\Vert^2_{\bs{\Sigma_2}^{\unaryminus 1}}}\;\mb{v_1},
& \quad &
\mathcal{P}_0^-(\bs{\hat\tau};\bs{\Sigma_2})=R^0+
\displaystyle\frac{\langle\bs{\hat\tau}-R^0,\mb{v_2}\rangle_{\bs{\Sigma_2}^{\unaryminus 1}}}{\Vert\mb{v_2}\Vert^2_{\bs{\Sigma_2}^{\unaryminus 1}}}\;\mb{v_2},\\[4mm]
\mathcal{P}_1^+(\bs{\hat\tau};\bs{\Sigma_2})=R^1+
\displaystyle\frac{\langle\bs{\hat\tau}-R^1,\mb{v_0}\rangle_{\bs{\Sigma_2}^{\unaryminus 1}}}{\Vert\mb{v_0}\Vert^2_{\bs{\Sigma_2}^{\unaryminus 1}}}\;\mb{v_0},
& &
\mathcal{P}_1^-(\bs{\hat\tau};\bs{\Sigma_2})=R^1+
\displaystyle\frac{\langle\bs{\hat\tau}-R^1,\mb{v_2}\rangle_{\bs{\Sigma_2}^{\unaryminus 1}}}{\Vert\mb{v_2}\Vert^2_{\bs{\Sigma_2}^{\unaryminus 1}}}\;\mb{v_2},\\[4mm]
\mathcal{P}_2^+(\bs{\hat\tau};\bs{\Sigma_2})=R^2+
\displaystyle\frac{\langle\bs{\hat\tau}-R^2,\mb{v_0}\rangle_{\bs{\Sigma_2}^{\unaryminus 1}}}{\Vert\mb{v_0}\Vert^2_{\bs{\Sigma_2}^{\unaryminus 1}}}\;\mb{v_0},
& &
\mathcal{P}_2^-(\bs{\hat\tau};\bs{\Sigma_2})=R^2+
\displaystyle\frac{\langle\bs{\hat\tau}-R^2,\mb{v_1}\rangle_{\bs{\Sigma_2}^{\unaryminus 1}}}{\Vert\mb{v_1}\Vert^2_{\bs{\Sigma_2}^{\unaryminus 1}}}\;\mb{v_1}.
\end{array}
\end{equation}
We need to know which of these projections are in $\partial U$. To this purpose, we first check if $\mathcal{P}_i^\pm(\bs{\hat\tau};\bs{\Sigma_2})$ lies on $P_2$ by verifying inequalities \eqref{eq:politopo}. If so, it holds $\mathcal{P}_i^\pm(\bs{\hat\tau};\bs{\Sigma_2})\in\partial U$ if and only if $l_i(\mathcal{P}_i^\pm(\bs{\hat\tau};\bs{\Sigma_2}))\geq 0,$
where $l_i(\bt)$ are the polynomials defining the lines $L_i,\ i=0,1,2$ (see equations \eqref{eq:lines}).

\subsection{The orthogonal projections on the arcs of the ellipse}\label{sc:projaE}
The projections $\mathcal{P}_E^i(\hat\bt;\bs{\Sigma_2}),\ i=1,\dots,k$ of the data point $\bs{\hat\tau}\in\RR^2$ on the ellipse $E$ are the stationary points of the squared Mahalanobis distance $\Vert\bs{\hat\tau}-\bt\Vert^2_{\bs{\Sigma_2}^{\unaryminus 1}}$ restricted to $E.$ 
Let $\nu$ be a Lagrange multiplier and consider the Lagrange function
$$
\Lambda(\bt,\nu;\bs{\hat\tau},\bs\Sigma_2)=\Vert\bs{\hat{\tau}}-\bt\Vert^2_{\bs\Sigma_2^{\unaryminus 1}}+\nu\,a(\bt)\, .
$$
Consequently $\mathcal{P}_E^i(\bs{\hat\tau};\bs\Sigma_2),\ i=1,\dots,k$ are the real stationary points of $\Lambda(\bt,\nu;\bs{\hat\tau},\bs\Sigma_2)$, i.e. the real solutions of the system
\begin{equation}\label{eq:MDequation}
\left\{\begin{array}{l}
\nabla_{\bt}\Lambda(\bt,\nu;\bs{\hat\tau},\bs\Sigma_2)=0\\
a(\bt)=0
\end{array}\right..
\end{equation}
This is a system of polynomial equations that can be solved through symbolic or numerical computation. In the first case, by using elimination theory (see \cite{clos}), we can reduce system \eqref{eq:MDequation} to a triangular polynomial system. In particular, we obtain a degree-$4$ equation in one variable, which admits a closed-form expression. The solutions of the system can then be derived through back--substitution. In doing so, however, we must be careful about issues of numerical stability of the solutions. From a numerical standpoint, in order to solve system \ref{eq:MDequation} we can use some software based on homotopy continuation (e.g. PHCpack \cite{Verschelde1999} or Bertini \cite{Bates2013}). In the following paragraph, we propose an alternate approach based on the parametric description of $E$ via trigonometric functions.

Recalling that $\alpha=\widehat{\mb{d_{10}}\mb{d_{20}}}$, we can state the following:

\begin{proposition}
Assume $\varphi\in[0,2\pi).$ Then
$\bt(\varphi)=(d_{10}\,\sin\varphi,d_{20}\,\sin(\varphi+\alpha))$
is a $1$--to--$1$ regular parametrization of $E.$
\end{proposition}
\begin{proof}
As the Cartesian equation of $E$ is $a(\bt)=0$, for any given $\bt\in E$ the vector $\mb{u}(\bt)=\tau_{20}\mb{d_{10}}-\tau_{10}\mb{d_{20}}$ satisfies
$$
\Vert\mb{u}(\bt)\Vert^2=
\Vert\tau_{20}\mb{d_{10}}-\tau_{10}\mb{d_{20}}\Vert^2= \Vert\mb{d_{10}}\wedge\mb{d_{20}}\Vert^2.
$$
In the $x$--plane, the vectors $\ast\mb{d_{10}},\ast\mb{d_{20}}$ are perpendicular to $\mb{d_{10}},\mb{d_{20}},$ respectively, therefore $\{\mb{d_{10}},\ast\mb{d_{10}}\}$ and $\{\mb{d_{20}},\ast\mb{d_{20}}\}$ are both orthogonal bases of $\RR^2.$
This implies that, for any $\bt\in E,$ there exists a unique angle $\varphi\in[0,2\pi)$ such that
\begin{equation}\label{eq:uphi}
\mb{u}(\bt)=-\ast(\mb{d_{10}}\wedge\mb{d_{20}})\;
(\tilde{\mb{d}}_{\mb{10}}\cos\varphi+\ast\tilde{\mb{d}}_{\mb{10}}\sin\varphi)\,,
\end{equation}
where $\tilde{\mb{d}}_{\mb{10}},\ast\tilde{\mb{d}}_{\mb{10}}$ are unit vectors. From the definition of $\mb{u}(\bt),$ we have
$$
\tau_{10}=-\frac{\langle\mb{u}(\bt),\ast\mb{d_{10}}\rangle}{\langle\mb{d_{20}},\ast\mb{d_{10}}\rangle}\,,\qquad\qquad
\tau_{20}=\frac{\langle\mb{u}(\bt),\ast\mb{d_{20}}\rangle}{\langle\mb{d_{10}},\ast\mb{d_{20}}\rangle}\,.
$$
By substituting \eqref{eq:uphi} in the above formulas, we obtain the following trigonometric parametrization of $E$:
\begin{equation*}
\left\{
\begin{array}{l}
\tau_{10}(\varphi)=\displaystyle-\frac{\ast(\mb{d_{10}}\wedge\mb{d_{20}})}{\langle\mb{d_{20}},\ast\mb{d_{10}}\rangle}\,
\left(\langle\tilde{\mb{d}}_{\mb{10}},\ast\mb{d_{10}}\rangle\cos\varphi+
\langle\ast\tilde{\mb{d}}_{\mb{10}},\ast\mb{d_{10}}\rangle\sin\varphi\right)\\[5mm]
\tau_{20}(\varphi)=\displaystyle\frac{\ast(\mb{d_{10}}\wedge\mb{d_{20}})}{\langle\mb{d_{10}},\ast\mb{d_{20}}\rangle}\,
\left(\langle\tilde{\mb{d}}_{\mb{10}},\ast\mb{d_{20}}\rangle\cos\varphi+
\langle\ast\tilde{\mb{d}}_{\mb{10}},\ast\mb{d_{20}}\rangle\sin\varphi\right)
\end{array}
\right.\,.
\end{equation*}
In the Euclidean plane, the identities $\langle\ast\mb{u},\ast\mb{v}\rangle=\langle\mb{u},\mb{v}\rangle$ and $\langle\mb{u},\ast\mb{v}\rangle=\ast(\mb{u}\wedge\mb{v})$ hold for any $\mb{u},\mb{v}\in\RR^2$ (see Appendix A of \cite{Compagnoni2013a}), therefore
$$
\tau_{10}(\varphi)=d_{10}\sin\varphi
$$
and
$$
\tau_{20}(\varphi)=
\ast(\tilde{\mb{d}}_{\mb{10}}\wedge\mb{d_{20}})\cos\varphi+
\langle\tilde{\mb{d}}_{\mb{10}},\mb{d_{20}}\rangle\sin\varphi=
d_{20}(\sin\alpha\cos\varphi+\cos\alpha\sin\varphi)=
d_{20}\sin(\varphi+\alpha)\,,
$$
where in the second equality we used the assumption $\ast(\mb{d_{10}}\wedge\mb{d_{20}})>0.$
\end{proof}
\noindent By substituting the above parametrization of $E$ into $\Vert\bs{\hat\tau}-\bt\Vert^2_{\bs\Sigma_2^{\unaryminus 1}}$ and differentiating with respect to $\varphi$, we obtain the following trigonometric equation in $\varphi\in [0,2\pi):$
\begin{equation}\label{eq:projell}
\left<\bs{\hat\tau}-\bt(\varphi)\,,\,(d_{10}\cos\varphi,d_{20}\cos(\varphi+\alpha))
\right>_{\bs\Sigma_2^{\unaryminus 1}}=0.
\end{equation}
For example, in the simplest case with $\bs\Sigma_2=\sigma^2 \mb{I},$ the equation is
\begin{equation*}\label{eq:ProjEllipse}
d_{10}\cos\left(\tilde\varphi-\frac{\alpha}{2}\right)
\left(\hat\tau_{10}-d_{10}\sin\left(\tilde\varphi-\frac{\alpha}{2}\right)\right)+
d_{20}\cos\left(\tilde\varphi+\frac{\alpha}{2}\right)
\left(\hat\tau_{20}-d_{20}\sin\left(\tilde\varphi+\frac{\alpha}{2}\right)\right)=0,
\end{equation*}
where we used the more symmetric variable $\tilde\varphi=\varphi+\frac{\alpha}{2}\in[0,2\pi).$
For any fixed setting of the sensors and the matrix $\bs\Sigma_2$, and for any TDOA measurements $\bs{\hat\tau}\in\RR^2$, solutions $\bar\varphi_i,\ i=1,\dots,k$ of equation \eqref{eq:projell} can be obtained through standard numerical algorithms (e.g. bisection or Newton--Rapson methods). Then, the relative orthogonal projections of $\bs{\hat\tau}$ on $E$ are $\mathcal{P}_E^i(\bs{\hat\tau};\bs\Sigma_2)=\bt(\varphi_i),\ i=1,\dots,k.$

Irrespective of the chosen resolution method, for any projection we have to finally check if $\mathcal{P}_E^i(\bs{\hat\tau};\bs\Sigma_2)$ lies on the right arcs of ellipse for the given model. In particular, if we are considering model $M,$ then we require $\mathcal{P}_E^i(\bs{\hat\tau};\bs\Sigma_2)\in\partial U\cap E,\ i=1,\dots,k.$ We showed in \cite{Compagnoni2013a} that 
\begin{equation}\label{eq:bordoU_i}
\partial(U_0\cup U_1\cup U_2)\cap E
\setminus\{T_0^\pm,T_1^\pm,T_2^\pm\} = C^+\cap E
\end{equation}
(see Figure \ref{fig:Ellisse-Politopo}). Since the set $\partial U\cap E$ is the complement in $E$ of \eqref{eq:bordoU_i}, we have
\begin{equation}\label{eq:bordoU}
\partial U\cap E = (C^-\cup C)\cap E.
\end{equation}
By definition, $\mathcal{P}_E^i(\bs{\hat\tau};\bs{\Sigma_2})\in E,\ i=1,\dots,k,$ thus $\mathcal{P}_E^i(\bs{\hat\tau};\bs\Sigma_2)\in \partial U\cap E$ if, and only if, $b(\mathcal{P}_E^i(\bs{\hat\tau};\bs\Sigma_2))\leq 0.$ On the other hand, if we are considering the models $M_j,\ j=0,1,2,$ then we have to check if $\mathcal{P}_E^i(\bs{\hat\tau};\bs\Sigma_2)\in\partial U_j\cap E,\ i=1,\dots,k.$ This holds if and only if inequality $l_j(\mathcal{P}_E^i(\bs{\hat\tau};\bs\Sigma_2))\geq 0$ is satisfied.

We conclude this subsection by discussing on the number $k$ of solutions of system \eqref{eq:MDequation}. For the case $\bs\Sigma_2=\sigma^2\bs I,$ this problem is known in the algebraic geometry literature (see \cite{Draisma2013,Friedland2014}) as the computation of the Euclidean distance degree of a variety (the ellipse $E$ in the present case). We remark that the knowledge of $k$ is crucial for the correct functioning of any numerical algorithm used for solving system \eqref{eq:MDequation}.

\begin{definition}\label{def:MDdegre}
The Mahalanobis distance degree (MDdegree) of the ellipse $E$ is the number of complex stationary points of the Lagrangian $\Lambda(\bt,\nu;\bs{\hat\tau};\bs{\Sigma_2})$
for a general $\hat\bt\in\CC^2.$ The real Mahalanobis Degree (rMD) of $E$ is the integer valued function $\kappa_{\bs\Sigma_2}:\RR^2\rightarrow\mathbb{N}$ that for any $\bs{\hat\tau}\in\RR^2$ gives the number of distinct real stationary points of $\Lambda(\bt,\nu;\bs{\hat\tau};\bs{\Sigma_2}).$ Finally, the Mahalanobis degree discriminant (MDdiscriminant) of $E$ is the locus $\mathcal{E}_{\bs\Sigma_2}\subset\CC^2$ of points $\bs{\hat\tau}\in\RR^2$ such that $\Lambda(\bt,\nu;\bs{\hat\tau};\bs{\Sigma_2})$ has at least two coinciding stationary points.
\end{definition}
\noindent The interested reader can find in \cite{Draisma2013,Friedland2014} the definition of the distance degree and discriminant for any given variety and the proofs of the following results.
As it is implicit in its definition, MDdegree does not depend on $\bs{\hat\tau}$ and $\bs\Sigma_2$ and for every ellipse it is equal to $4.$ The MDdiscriminant $\mathcal{E}_{\bs\Sigma_2}$ is an astroid, i.e. an algebraic singular curve of degree six whose real part is of the type drawn in Figure \ref{fig:Evoluta}.
\begin{figure}[htb]
\begin{center}
\resizebox{11cm}{!}{
  \includegraphics[bb=92 263 519 528]{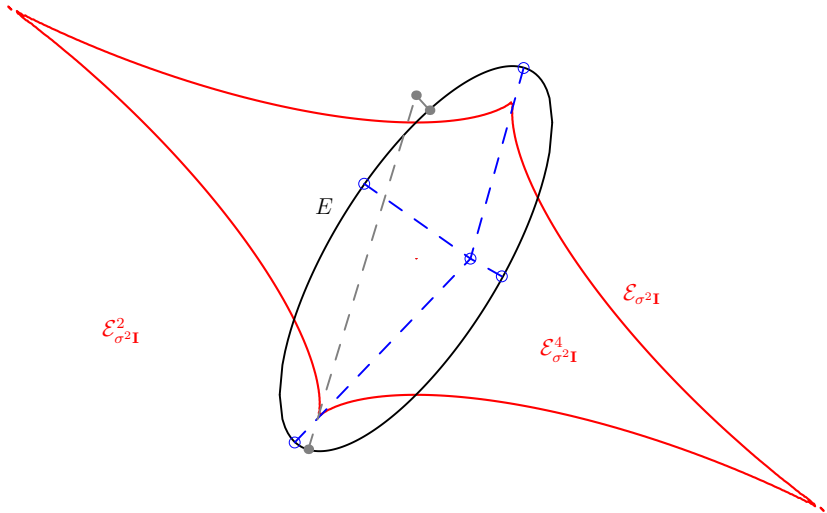}}\\
  \caption{The MDdiscriminant $\mathcal{E}_{\bs\Sigma_2}$ of the ellipse $E$ when $\bs\Sigma_2=\sigma^2\mb{I}$ and the regions $\mathcal{E}^2_{\bs\Sigma_2},\mathcal{E}^4_{\bs\Sigma_2}$, for the configuration of the receiver $\m{0}=(0,0)^T,\ \m{1}=(2,0)^T$ and $\m{2}=(2,2)^T$. If $\bs{\hat{\tau}}\in\mathcal{E}^4_{\bs\Sigma_2}$ there are $4$ projections on $E$, if $\bs{\hat{\tau}}\in\mathcal{E}^2_{\bs\Sigma_2}$ there are $2$ projections and finally for $\bs{\hat{\tau}}\in\mathcal{E}_{\bs\Sigma_2}$ there are $3$ distinct projections.}
  \label{fig:Evoluta}
\end{center}
\end{figure}

\noindent Although the shape of the astroid depends on the covariance matrix, its topological properties are invariant and in particular $\mathcal{E}_{\bs\Sigma_2}$ subdivides the $\tau$--plane into two connected and disjoint open regions: the exterior and the interior of $\mathcal{E}_{\bs\Sigma_2},$ that we name $\mathcal{E}_{\bs\Sigma_2}^2$ and $\mathcal{E}_{\bs\Sigma_2}^4,$ respectively. The function $\kappa_{\bs\Sigma_2}$ is constant on each of these regions and in particular we have:
\begin{equation}
\kappa_{\bs\Sigma_2}(\hat\bt)=\left\{
\begin{array}{lcl}
2 & \qquad\text{if, and only if,}\qquad & \bs{\hat\tau}\in\mathcal{E}_{\bs\Sigma_2}^2\\
3 & \text{if, and only if,} & \bs{\hat\tau}\in\mathcal{E}_{\bs\Sigma_2}\\
4 & \text{if, and only if,} & \bs{\hat\tau}\in\mathcal{E}_{\bs\Sigma_2}^4\\
\end{array}\right.\,.
\end{equation}
Notice that $\kappa_{\bs\Sigma_2}(\bs{\hat\tau})$ is exactly the number $k$ of orthogonal projections $\mathcal{P}_E^i(\bs{\hat\tau};\bs\Sigma_2)$ of the point $\bs{\hat\tau}\in\RR^2$ on the ellipse $E.$ In appendix \ref{app:evoluta} we include the source code (written in Singular \cite{DGPS}) for computing the Cartesian equation $F_{\bs\Sigma_2}(\bt)=0$ of $\mathcal{E}_{\bs\Sigma_2},$ once the sensors positions and the covariance matrix have been set. As an example, the polynomial defining the curve $\mathcal{E}_{\sigma^2\bs{I}}$ in Figure \ref{fig:Evoluta} is:
$$
F_{\sigma^2\bs{I}}(\bt)=
\tau_1^6+6\tau_1^5\tau_2+18\tau_1^4\tau_2^2+32\tau_1^3\tau_2^3+36\tau_1^2\tau_2^4+
24\tau_1\tau_2^5+8\tau_2^6+48\tau_1^4-24\tau_1^3\tau_2-
$$
$$
-588\tau_1^2\tau_2^2-696\tau_1\tau_2^3-132\tau_2^4+1200\tau_1^2+
2400\tau_1\tau_2+2400\tau_2^2-8000.
$$
As $\bs{0}=(0,0)^T\in\mathcal{E}_{\bs\Sigma_2}^4$, a point $\bs{\hat\tau}\in\RR^2$ lies on $\mathcal{E}_{\bs\Sigma_2}^4$ (respectively $\mathcal{E}_{\bs\Sigma_2}^2$) if, and only if, $F_{\bs\Sigma_2}(\bs{\hat\tau})F_{\bs\Sigma_2}(\bs{0})>0$ (respectively $F_{\bs\Sigma_2}(\bs{\hat\tau})F_{\bs\Sigma_2}(\bs{0})<0$).

\section{Asymptotic statistical inference}\label{sc:ASI}
A crucial point in parametric statistics is the evaluation of the accuracy of parameter estimation procedures. 
The precise description of the parametric models $M,M_0,M_1,M_2$ given in Section \ref{sc:StatMod} allows us to use the tools of Information Geometry \cite{Amari2000} for analyzing this aspect of TDOA--based localization. In Section \ref{sc:sourceest} we laid down the basis for the solution of MLE, that is the optimal estimation from the statistical point of view. Now we focus on the analysis of MLE efficiency. In Section \ref{sec:validation} we will explicitly implement the MLE and validate our theoretical analysis.

Notice that the analysis proposed in this manuscript is inherently local, based on the differential geometry properties of the restricted models. In particular, this means that here we do not take into account the ambiguities in the localization, which we described in Section \ref{sc:TDOAspace}. What we do is study an ideal situation where we know the region $\Omega,\Omega_1,\Omega_2,\Omega_3$ where the source lies, and therefore we know which TDOA model to consider. At the end of Section \ref{sec:validation}, we will go back to the more realistic scenario in which we have no a-priori knowledge on the source position.

We finally remark that the analytic evaluation of the accuracy of source localization is not a novel idea per se. See \cite{Ho2012,So2013} for other examples
that do not rely on Information Geometry. This manuscript, however, goes further in that investigation as it proposes in Subsection \ref{sc:assASI} a study on the reliability of the asymptotic analysis.

\subsection{Asymptotic mean square error and bias of MLE}
From Theorem \ref{th:curvedexpfam}, we know that each model $M,M_0,M_1,M_2$ is a $(2,2)$ curved exponential family. As $\bs{\Sigma_2}$ is known, we only have to specify the parameters $\bs{\theta}(\x)=\bs{\tau_2}(\x)$. It is well-known that for any statistical manifold, there exists a natural Riemannian metric, the so called \emph{Fisher metric}. If we use $E_{\x}[\cdot]$ to denote the expectation value with respect to the distribution $p(\bs{\hat{\tau}};\bs{\theta}(\x))$, then at any point $\bs{\theta}(\x)$ the metric is given by the Fisher information matrix
$$
\mb{G}(\x)=E_{\x}[\nabla_{\x}\ell(\bs{\hat{\tau}};\bs{\theta}(\x))^T\,\nabla_{\x}\ell(\bs{\hat{\tau}};\bs{\theta}(\x))]=
\int_{\RR^2}\nabla_{\x}\ell(\bs{\hat{\tau}};\bs{\theta}(\x))^T\,\nabla_{\x}\ell(\bs{\hat{\tau}};\bs{\theta}(\x))\,p(\bs{\hat{\tau}};\bs{\theta}(\x))\mb{d}\bs{\hat\tau},
$$
where $\nabla_{\x}\ell(\bs{\hat{\tau}};\bs{\theta}(\x))=(\partial_x\ell(\bs{\hat{\tau}};\bs{\theta}(\x)),\partial_y\ell(\bs{\hat{\tau}};\bs{\theta}(\x))).$ In the integral appears the function
$$
\ell(\bs{\hat{\tau}};\bs{\theta}(\x))=\log(p(\bs{\hat{\tau}};\bs{\theta}(\x)))=
C(\bs{\hat\tau})+\sum_{i=1}^n \theta_i(\x)F_i(\bs{\hat\tau})-\psi(\bs{\theta}(\x)),
$$
where
$$
C(\bs{\hat\tau})=-\frac{1}{2}\,\Vert\bs{\hat{\tau}}\Vert_{\bs{\Sigma_2}^{\unaryminus 1}}^2\,,\qquad
F_i(\bs{\hat{\tau}})=(\bs{\Sigma_2}^{-1}\bs{\hat{\tau}})_i\,,\qquad
\psi(\bs{\theta}(\x))=\frac{1}{2}\,\Vert\bs\theta(\x)\Vert_{\bs{\Sigma_2}^{\unaryminus 1}}^2
+\log\sqrt{(2\pi)^n|\bs{\Sigma_2}|}\,.
$$

In order to explicitly obtain the Fisher matrix, we first compute the Jacobian matrix of $\bs{\theta}(\x):$
$$
\mb{J}(\bs{\theta}(\x))=\left( \begin{array}{c}
\mb{\tilde{d}_1}(\x)-\mb{\tilde{d}_0}(\x)\\
\mb{\tilde{d}_2}(\x)-\mb{\tilde{d}_0}(\x)
\end{array} \right).
$$
After some straightforward computations we obtain
$$
\mb{G}(\x)=\mb{J}(\bs{\theta}(\x))^T\bs{\Sigma_2}^{-1}\mb{J}(\bs{\theta}(\x)).
$$
Let us also compute the Hessian matrices of the components of $\bs{\theta}(\x)$ with respect to the parameters $\x$:
$$
\mb{Hs}(\theta_{i}(\x))=\left( \begin{array}{cc}
\frac{\PM{\mb{\tilde{d}_i},\mb{e_2}}^2}{d_i} &
-\frac{\PM{\mb{\tilde{d}_i},\mb{e_1}}
\PM{\mb{\tilde{d}_i},\mb{e_2}}}{d_i}\\
-\frac{\PM{\mb{\tilde{d}_i},\mb{e_1}}
\PM{\mb{\tilde{d}_i},\mb{e_2}}}{d_i} &
\frac{\PM{\mb{\tilde{d}_i},\mb{e_1}}^2}{d_i}\\
\end{array} \right)-
\left( \begin{array}{cc}
\frac{\PM{\mb{\tilde{d}_0},\mb{e_2}}^2}{d_0} &
-\frac{\PM{\mb{\tilde{d}_0},\mb{e_1}}
\PM{\mb{\tilde{d}_0},\mb{e_2}}}{d_0}\\
-\frac{\PM{\mb{\tilde{d}_0},\mb{e_1}}
\PM{\mb{\tilde{d}_0},\mb{e_2}}}{d_0} &
\frac{\PM{\mb{\tilde{d}_0},\mb{e_1}}^2}{d_0}\\
\end{array} \right),
$$
for $i=1,2.$ Now, we can state the main result of the Section on MLE. As we saw in Section \ref{sc:MLE}, for any set of measurements $\bs{\hat\tau}$ and for each model we have an MLE estimate
$\bs{\bar{\tau}}(\bs{\hat{\tau}};\bs{\Sigma_2})$
and a corresponding source position $\mb{\bar{x}}(\bs{\hat{\tau}};\bs{\Sigma_2})=\bs{\tau_2}^{-1}(\bs{\bar{\tau}}(\bs{\hat{\tau}};\bs{\Sigma_2}))$.
\begin{proposition}\label{prop:asyan}
Given a source at $\x,$ the (local) asymptotic mean square error of $\mb{\bar{x}}(\bs{\hat{\tau}};\bs{\Sigma_2})$ is equal to
\begin{equation}\label{eq:AMSEMLE}
E_{\x}[(\mb{\bar{x}}(\bs{\hat{\tau}};\bs{\Sigma_2})-\mb{x})^T(\mb{\bar{x}}(\bs{\hat{\tau}};\bs{\Sigma_2})-\mb{x})]=\mb{G}(\x)^{-1}.
\end{equation}
The (local) first order bias of $\mb{\bar{x}}(\bs{\hat{\tau}};\bs{\Sigma_2})$ is
\begin{equation}\label{eq:ABMLE}
E_{\x}[\mb{\bar{x}}(\bs{\hat{\tau}};\bs{\Sigma_2})-\x]=-\frac{1}{2}\,\mb{b}(\x),
\end{equation}
where
\begin{equation}\label{eq:ABMLE2}
\mb{b}(\x)=
\left(\rm{Tr}(\mb{Hs}(\theta_1(\x))\mb{G}(\x)^{-1}),\
\rm{Tr}(\mb{Hs}(\theta_2(\x))\mb{G}(\x)^{-1})\right)
\cdot(\mb{J}(\bs{\theta}(\x))^{-1})^T.
\end{equation}
\end{proposition}
\begin{proof}
The formula for the asymptotic mean square error of a consistent estimator is given in Theorem 4.3 of \cite{Amari2000}. Since MLE is consistent and asymptotically efficient, such formula reduces to \eqref{eq:AMSEMLE}.\footnote{Actually, this way we obtain the Cramer--Rao lower bound.} Moreover, by considering the higher order asymptotic theory, one can compute the expected value in \eqref{eq:ABMLE}. It is just a matter of computation to verify that \eqref{eq:ABMLE2} is equivalent to formula 4.42 in \cite{Amari2000}.
\end{proof}
\noindent In the case of $\bs{\Sigma_2}=\sigma^2\mb{I},$ we can explicitly compute \eqref{eq:AMSEMLE}:
$$
\mb{G}(\x)^{-1}=
\frac{\sigma^2}{\vert J(\bs{\eta}(\x))\vert^2}\;\sum_{i=1}^2
\left( \begin{array}{cc}
\PM{\mb{\tilde{d}_i}-\mb{\tilde{d}_0},\mb{e_2}}^2 &
-\PM{\mb{\tilde{d}_i}-\mb{\tilde{d}_0},\mb{e_1}}
\PM{\mb{\tilde{d}_i}-\mb{\tilde{d}_0},\mb{e_2}}\\
-\PM{\mb{\tilde{d}_i}-\mb{\tilde{d}_0},\mb{e_1}}
\PM{\mb{\tilde{d}_i}-\mb{\tilde{d}_0},\mb{e_2}} &
\PM{\mb{\tilde{d}_i}-\mb{\tilde{d}_0},\mb{e_1}}^2
\end{array} \right).
$$

The a priori knowledge of the bias is very interesting from the point of view of applications, because it allows us to define the bias-corrected MLE as
\begin{equation}
\label{eq:loc_bc}
\mb{\bar{x}}_{bc}(\bs{\hat{\tau}};\bs{\Sigma_2})=\mb{\bar{x}}(\bs{\hat{\tau}};\bs{\Sigma_2})+\frac{1}{2}\mb{b}(\x).
\end{equation}
However, an exact compensation of the bias involves the knowledge of the true source location $\x$, which is obviously unknown in real context. At best, one can compute the bias at the estimated source location, and so the bias-compensated estimate becomes
\begin{equation}
\label{eq:loc_bc2}
\mb{\bar{x}}_{bc}(\bs{\hat{\tau}};\bs{\Sigma_2})=\mb{\bar{x}}(\bs{\hat{\tau}};\bs{\Sigma_2})+\frac{1}{2}\mb{b}(\mb{\bar{x}}(\bs{\hat{\tau}};\bs{\Sigma_2})).
\end{equation}
An inexact knowledge of the source location introduces an error in the bias prediction. With some preliminary experiments we tested that \eqref{eq:loc_bc2} improves the estimate of the distance of the source from the reference sensor, while the variance of the direction of arrival increases. This suggests that an improved accuracy could be achieved by using the estimation of the distance and of the angle coming from bias-corrected and non-corrected estimations, respectively. However, we leave a deeper analysis of such a problem for future developments.

\subsection{Assessment of the asymptotic analysis}\label{sc:assASI}
The asymptotic error analysis described in the previous Section gives a correct evaluation of the error relative to the source position only under certain conditions. Indeed, it works well in the regions of the statistical manifold where the curvature is not too high (see Section 4.5 of \cite{Amari2000}). For each one of the four models $M,M_0,M_1,M_2,$ this remains true if we keep away from their boundaries.

In this Section we propose a method for evaluating the reliability of the asymptotic analysis. The  starting observation is that the asymptotic error analysis is essentially based on taking the first non-trivial orders in the Taylor expansions of the expectation values \eqref{eq:AMSEMLE} and \eqref{eq:ABMLE}, respectively. A first approach is to consider the rest of such approximations, for example, by taking the Lagrange remainders of the respective Taylor polynomials. However, preliminary tests indicate that, in doing so, we typically overestimate the errors caused by the low-order approximations.

More realistically, we can estimate the error in the asymptotic approximation by computing the next order in the series expansion of \eqref{eq:AMSEMLE} and \eqref{eq:ABMLE}. In order to do so the key identity is
$$
E_{\x}[f(\bs{\hat\tau})]=\int_{\RR^2}f(\bs{\hat{\tau}})\,p(\bs{\hat\tau};\bs{\theta}(\x))\,\mb{d}\bs{\hat\tau}=\left.\exp\left(\frac{1}{2}\nabla\,\bs{\Sigma_2}\nabla^T\right)f(\bs{\hat\tau})\right|_{\bs{\hat\tau}=\bs{\theta}(\x)},
$$
where the exponential differential operator corresponds to the power series
$$
\exp\left(\frac{1}{2}\nabla\,\bs{\Sigma_2}\nabla^T\right)=
\sum_{n=0}^\infty\frac{1}{2^n n!}\left(\sum_{i,j=1}^2\bs{\Sigma_2}_{ij}\,\frac{\partial}{\partial\hat\tau_i}\,\frac{\partial}{\partial\hat\tau_j}\right)^n.
$$
For the sake of simplicity, here we focus on the analysis of \eqref{eq:AMSEMLE} under the assumption $\bs{\Sigma_2}=\sigma^2\mb{I}.$ We have
$$
f(\bs{\hat\tau})=\left((\mb{\bar{x}}(\bs{\hat{\tau}};\bs{\Sigma_2})-\mb{x})^T(\mb{\bar{x}}(\bs{\hat{\tau}};\bs{\Sigma_2})-\mb{x})\right)_{ij},\qquad i,j=1,2.
$$
If $\bs{\hat\tau}$ lies on the interior of a given model, the MLE is simply the inverse map, i.e. $\mb{\bar{x}}(\bs{\hat{\tau}};\bs{\Sigma_2})=\bs{\tau_2}^{-1}(\bs{\hat\tau}).$ Therefore, it is a matter of computation to obtain the first order correction $\bs{\Delta}(\x)$ to the expectation value \eqref{eq:AMSEMLE}. In order to simplify the explanation, in the formula we use the following multi-index notation:
$$
D^{(p,q)}\bar{x}_i=\left(\left.\frac{\partial^{p+q}}{\partial\hat\tau_1^p\,\partial\hat\tau_2^q}\,
\bs{\tau_2}^{-1}(\bs{\hat\tau})\right|_{\bs{\hat\tau}=\bs{\theta}(\x)}\right)_i
\quad\text{for}\quad i=1,2.
$$
This way, we arrive to
\begin{equation}\label{eq:remainder}
\bs{\Delta}(\x)_{ij}=\left(E_{\x}[(\mb{\bar{x}}(\bs{\hat{\tau}};\bs{\Sigma_2})-\mb{x})^T(\mb{\bar{x}}(\bs{\hat{\tau}};\bs{\Sigma_2})-\mb{x})]-\mb{G}(\x)^{-1}\right)_{ij}=
\end{equation}
$$
\frac{\sigma^4}{4}\left(3\left(D^{(2,0)}\bar{x}_i D^{(2,0)}\bar{x}_j+
D^{(0,2)}\bar{x}_i D^{(0,2)}\bar{x}_j\right)
+D^{(2,0)}\bar{x}_iD^{(0,2)}\bar{x}_j+D^{(0,2)}\bar{x}_iD^{(2,0)}\bar{x}_j+
4D^{(1,1)}\bar{x}_iD^{(1,1)}\bar{x}_j+\right.
$$
$$
+2D^{(1,0)}\bar{x}_i\left(D^{(3,0)}\bar{x}_j+D^{(1,2)}\bar{x}_j\right)+
2D^{(0,1)}\bar{x}_i\left(D^{(0,3)}\bar{x}_j+D^{(2,1)}\bar{x}_j\right)+
$$
$$
\left.+2D^{(1,0)}\bar{x}_j\left(D^{(3,0)}\bar{x}_i+D^{(1,2)}\bar{x}_i\right)+
2D^{(0,1)}\bar{x}_j\left(D^{(0,3)}\bar{x}_i+D^{(2,1)}\bar{x}_i\right)\right)
+o(\sigma^5).
$$

In Section \ref{sec:validation} we will validate the formula through simulations and we will suggest how \eqref{eq:remainder} can be used for the evaluation of the accuracy of the asymptotic error analysis given in Proposition \ref{prop:asyan}.

\section{MLE algorithm implementation and validation}\label{sec:validation}

This Section implements the MLE localization technique and the asymptotic statistical analysis in Sections \ref{sc:sourceest} and \ref{sc:ASI}, respectively. In particular, we will validate them through a set of simulations.
Finally, we conclude by giving offering some comments for the reader who might be interested in adopting the described algorithm in a real scenario.

\subsection{The solutions of the MLE}\label{sec:MLEfinale}

As discussed in Section \ref{sc:sourceest}, we have a different MLE for models $M$ and $M_i,~i=1,2,3$.
In all these cases we are interested in computing the estimate $\mb{\bar{x}}(\bs{\hat\tau};\bs{\Sigma_2})$ from the data vector $\bs{\hat\tau}$. We begin with the model $M$, therefore we assume that $\x\in\Omega.$
\begin{algorithm}[H]
\caption{MLE algorithm for model $M$}
\label{alg:MLE_M}
\begin{algorithmic}[1]
\REQUIRE{TDOA measurements $\bs{\hat\tau}\in\RR^2,$ covariance matrix $\bs{\Sigma_2}$}
\STATE{Check if $\bs{\hat\tau}$ lies in $U$ by verifying the (strict) inequalities \eqref{eq:politopo} and \eqref{eq:U}. If so $\bs{\bar\tau}(\bs{\hat\tau};\bs{\Sigma_2})=\bs{\hat\tau}$ and go to Line  \ref{line5}.}
\STATE{Compute the orthogonal projections $\mathcal{P}_i^\pm(\bs{\hat\tau};\bs{\Sigma_2}),\ i=0,1,2,$ using formulas \eqref{eq:projfacets} and $\mathcal{P}_E^j(\bs{\hat\tau};\bs{\Sigma_2}),\ j=1,\dots,k,$ by solving system \eqref{eq:MDequation}.}
\STATE{Evaluate which of the projections lie on $\partial U.$ For $\mathcal{P}_i^\pm(\bs{\hat\tau};\bs{\Sigma_2}),\ i=0,1,2,$ one must check inequalities \eqref{eq:politopo} and $l_i(\mathcal{P}_i^\pm(\bs{\hat\tau};\bs{\Sigma_2}))\geq 0.$ For $\mathcal{P}_E^j(\bs{\hat\tau};\bs{\Sigma_2}),\ j=1,\dots,k,$ one must verify $b(\mathcal{P}_E^j(\bs{\hat\tau};\bs{\Sigma_2}))\leq 0.$}
\STATE{For every projection on $\partial U$ and for each vertex $R^i,\ i=0,1,2,$ computer the Mahalanobis distance from $\bs{\hat\tau}.$ The MLE solution $\bs{\bar\tau}(\bs{\hat\tau};\bs{\Sigma_2})$ corresponds to the point where such distance is minimum.}
\STATE\label{line5}{The solution of the MLE in the $x$--plane is $\mb{\bar{x}}(\bs{\hat\tau};\bs{\Sigma_2})=\x_+(\bs{\bar\tau}(\bs{\hat\tau};\bs{\Sigma_2})).$}
\RETURN $\mb{\bar{x}}(\bs{\hat\tau};\bs{\Sigma_2})$
\end{algorithmic}
\end{algorithm}
We must pay attention to the interpretation of the results regarding the source position. Indeed, by taking the closest point to $\bs{\hat\tau}$ on $\partial U,$ we are actually considering the {\em compactification}
of the model $M.$ We have the following cases.
\begin{itemize}
\item
If $\bs{\bar\tau}(\bs{\hat\tau};\bs{\Sigma_2})$ is one of the $\mathcal{P}_i^\pm(\bs{\hat\tau};\bs{\Sigma_2}),\ i=0,1,2,$ then $\mb{\bar{x}}(\bs{\hat\tau};\bs{\Sigma_2})$ lies on the degeneracy locus $D$.
\item
If $\bs{\bar\tau}(\bs{\hat\tau};\bs{\Sigma_2})$ is one of the $\mathcal{P}_E^j(\bs{\hat\tau};\bs{\Sigma_2}),\ j=1,\dots,k,$ then we have to take the extension of the inverse map $\mb{x}_+(\bt)$ with value in the projective plane, because $\lambda_+(\bt)$ is not defined on $\partial U\cap E.$ It follows that $\mb{\bar{x}}(\bs{\hat\tau};\bs{\Sigma_2})$ is the ideal point with homogeneous coordinates $(\mb{v}(\mathcal{P}_E^j(\bs{\hat\tau};\bs{\Sigma_2})):0)\in\PP_\RR^2$. In this case, the vector $\mb{v}(\mathcal{P}_E^j(\bs{\hat\tau};\bs{\Sigma_2}))$ should be interpreted as the localization direction of a very far source, in a situation where even a very small noise on the TDOA measurements hinders to
estimate the distance of the source from the sensors.
\item
If $\bs{\bar\tau}(\bs{\hat\tau};\bs{\Sigma_2})=R^i,$ then $\mb{\bar{x}}(\bs{\hat\tau};\bs{\Sigma_2})=\m{i},\ i=0,1,2.$
\end{itemize}

A similar MLE algorithm can be defined for each model $M_i,\ i=0,1,2$. In these cases, we are assuming that $\x\in\Omega_i.$
\begin{algorithm}[H]
\caption{MLE algorithm for model $M_i$}
\label{alg:MLE_Mi}
\begin{algorithmic}[1]
\setalglineno{1}
\REQUIRE{TDOA measurements $\bs{\hat\tau}\in\RR^2,$ covariance matrix $\bs{\Sigma_2}$}
\STATE{Check if $\hat\bt$ lies on $U_i$ by verifying the (strict) inequalities \eqref{eq:politopo} and \eqref{eq:Ui}. If so $\bs{\bar\bt_i}(\bs{\hat\tau};\bs{\Sigma_2})=\bs{\hat\tau}$ and go to Line \ref{line5i}.
}
\STATE{Compute the orthogonal projections $\mathcal{P}_i^\pm(\bs{\hat\tau};\bs{\Sigma_2})$ using formulas \eqref{eq:projfacets} and $\mathcal{P}_E^j(\bs{\hat\tau};\bs{\Sigma_2}),\ j=1,\dots,k,$ by  solving system \eqref{eq:MDequation}.}
\STATE{Evaluate which of the projections lie on $\partial U_i.$ For $\mathcal{P}_i^\pm(\bs{\hat\tau};\bs{\Sigma_2})$ one must check inequalities \eqref{eq:politopo} and $l_i(\mathcal{P}_i^\pm(\bs{\hat\tau};\bs{\Sigma_2}))\geq 0.$ For $\mathcal{P}_E^j(\bs{\hat\tau};\bs{\Sigma_2}),\ j=1,\dots,k,$ one must verify $l_i(\mathcal{P}_E^j(\bs{\hat\tau};\bs{\Sigma_2}))\geq 0.$}
\STATE{For every projection on $\partial U_i$ and for the vertex $R^i$ calculate the Mahalanobis distance from $\bs{\hat\tau}.$ The MLE solution $\bs{\bar\tau_i}(\bs{\hat\tau};\bs{\Sigma_2})$ corresponds to the point where such distance is minimum.}
\STATE\label{line5i}{The solution of the MLE in the $x$--plane is $\mb{\bar{x}_i}(\bs{\hat\tau};\bs{\Sigma_2})=\x_-(\bs{\bar\tau_i}(\bs{\hat\tau};\bs{\Sigma_2})).$}
\RETURN $\mb{\bar{x_i}}(\bs{\hat\tau};\bs{\Sigma_2})$
\end{algorithmic}
\end{algorithm}
Similar remarks to those offered for the model $M$ hold true in this case as well. We finally recommend to be careful about the numerical stability of the solutions of the quadratic equation $a(\bt)\lambda_{\pm}(\bt)^2+2b(\bt)\lambda_{\pm}(\bt)+c(\bt)=0$ contained in Section \ref{sc:TDOAspace}. This can become an issue especially when $\bt$ is close to the ellipse $E$ and so $a(\bt)\cong 0$. In this case the formula \eqref{eq:lambda} is ill conditioned. See for example \cite{higham2002} as a reference book on this topic.

\subsection{Simulative results and comparison with Asymptotical Statistical Inference}\label{sc:simulations}
In this subsection we show the experimental results about source localization based on Algorithms \ref{alg:MLE_M} and \ref{alg:MLE_Mi}. We evaluate them through asymptotic analysis and Monte Carlo simulations.
\subsubsection{Setup}
The sensors are deployed as in Figure \ref{fig:spazio-x}, i.e. $\mathbf{m}_0 = (0,0)^T$, $\mathbf{m}_1=(2,0)^T$ and $\mathbf{m}_2=(2,2)^T$, and the reference sensor is $\mathbf{m}_0$.
The zero-mean noise added to the measurements has a standard deviation of $\sigma=0.005~\textrm{m}$. Notice that both root mean square error and bias are proportional to the noise variance, and therefore results do not lose generality due to the choice of a specific value of $\sigma$.
Sources were placed on a regular grid centered around the center of gravity of the sensors. In particular, the $x$ and $y$ coordinates range from $-2.67~\textrm{m}$ to $5.33~\textrm{m}$ and from $-3.33~\textrm{m}$ to $5.33~\textrm{m}$, respectively, resulting in a total number of $4225$ test source locations.
For each source location, 500 Monte Carlo simulations of Algorithms \ref{alg:MLE_M} and \ref{alg:MLE_Mi} have been performed. Sample estimates $\bs{\hat\Sigma}(\x)$ and $\mb{\hat{B}}(\x)$ of the covariance matrix and the bias, respectively, are then computed.
In order to discriminate the distance and angular error components on the source location, the projections of $\bs{\hat\Sigma}(\x)$ and $\mb{\hat{B}}(\x)$ on the eigenvectors of the matrix $\mb{G}^{-1}(\mathbf{x})$ are computed. Indeed, for sufficiently distant sources, the eigenvector related to the largest eigenvalue approximately coincides with the direction of the source, as seen from the array (in the following called radial direction). Due to the dynamic range of the error along the radial component, we adopt a logarithmic transformation of the component of the covariance matrix (predicted or estimated) along the radial direction.
\subsubsection{Root Mean Square Error}
Figure \ref{fig:theoretical_rmse} shows the mean square error on localization predicted by $\mb{G}^{-1}(\mathbf{x})$ (first row) and the simulated one $\bs{\hat\Sigma}(\x)$ (second row), for the component along the radial direction (a) and the orthogonal one (b). Notice that the asymptotic prediction is quite accurate over the considered region, except for the areas surrounding the degeneracy locus \eqref{eq:degloc}, i.e. the half-lines $r_i^{\pm},~i=0,1,2$ prolongations of the segments joining the sensor locations.
\begin{figure}[htb]
\begin{center}
  \resizebox{15cm}{!}{
  \includegraphics*[bb=0 0 1217 592]{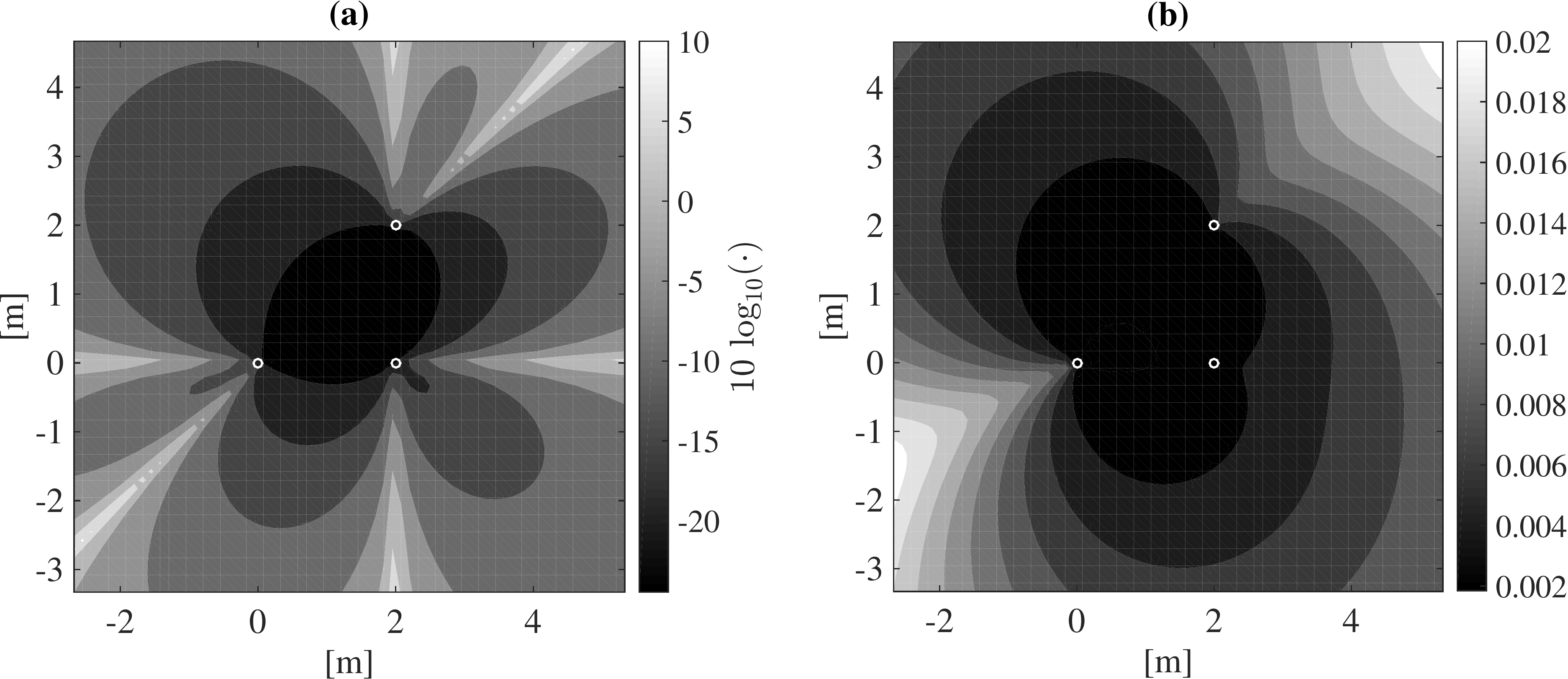}}
\\[5mm]
  \resizebox{15cm}{!}{
  \includegraphics*[bb=0 0 1217 592]{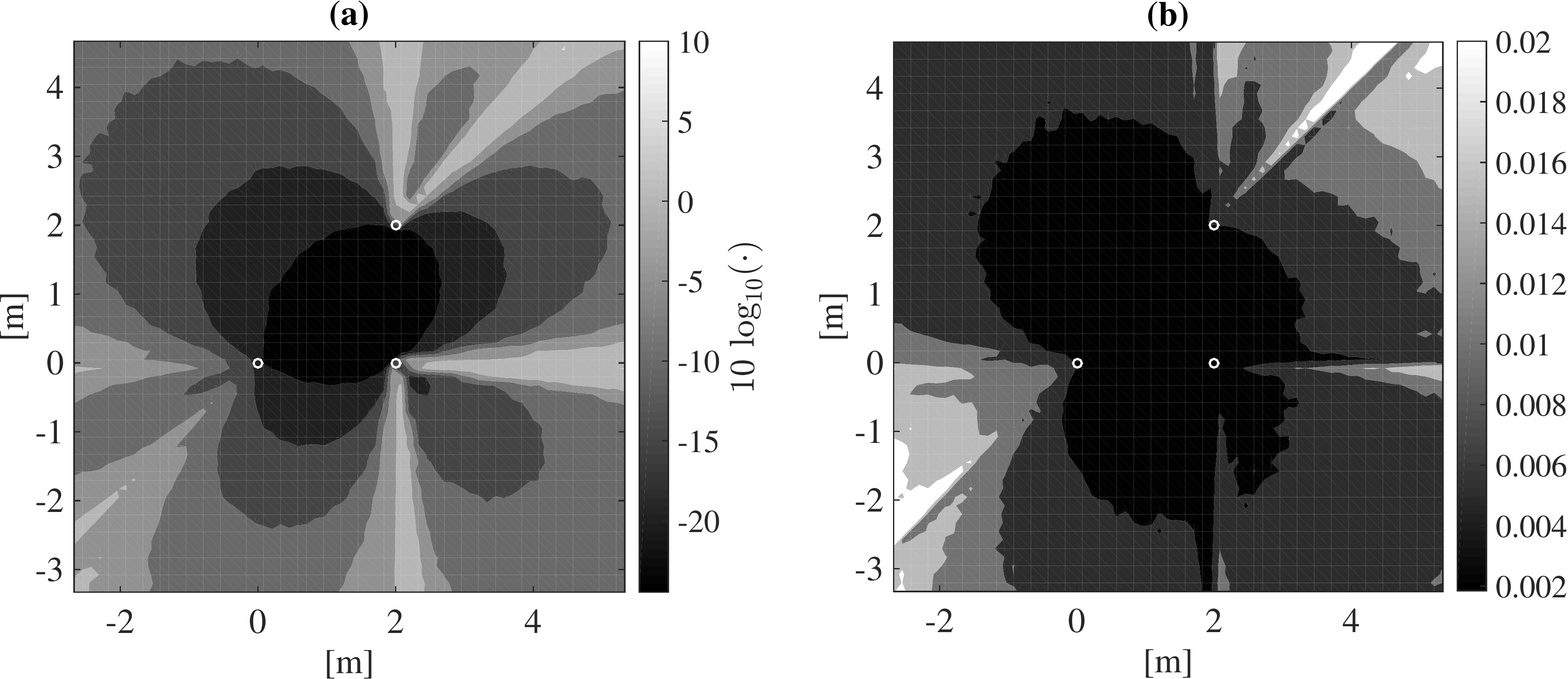}}
\caption{\label{fig:theoretical_rmse}Projection of $\mb{G}^{-1}(\mathbf{x})$ (first row) and  $\bs{\hat\Sigma}(\x)$ (second row) along the eigenvectors of $\mb{G}^{-1}(\mathbf{x})$. Projections along the eigenvectors related to the maximum (a) and minimum (b) eigenvalue of $\mb{G}^{-1}(\mathbf{x})$ are shown. Results are expressed on a logarithmic scale (a) and in meters (b).}
\end{center}
\end{figure}

With the aim of assessing the accuracy of the asymptotical estimate, in Figure \ref{fig:rem_analysis} we compare the remainder $\bs{\Delta}(\x)$ computed according to Equation \eqref{eq:remainder} and the difference $\bs{\hat\Delta}(\x)=\bs{\hat\Sigma}(\x)-\mb{G}^{-1}(\mathbf{x})$ between simulated and predicted RMSE. The solid lines are level curves of $\bs{\Delta}(\x),$ while the colormap represents $\bs{\hat\Delta}(\x).$ We can observe a pretty good match between the two. Such comparison suggests that we can use $\bs{\Delta}(\x)$ as a reliability certificate for the asymptotic prediction of RMSE given by $\mb{G}^{-1}(\mathbf{x})$. For example, we could define the trusted region of $\mb{G}^{-1}(\mathbf{x})$ as the one where $\bs{\Delta}(\x)$ takes value below a suitable threshold.

The availability of a method to predict the RMSE is important in applications where different accuracy is required in different regions. However, a quantitative estimate of $\bs{\Delta}(\x)-\bs{\hat\Delta}(\x)$ in the most general case is beyond the scope of the manuscript and needs further investigation.
\begin{figure}[htb]
\begin{center}
\resizebox{15cm}{!}{
  \includegraphics*[bb=0 0 1217 592]{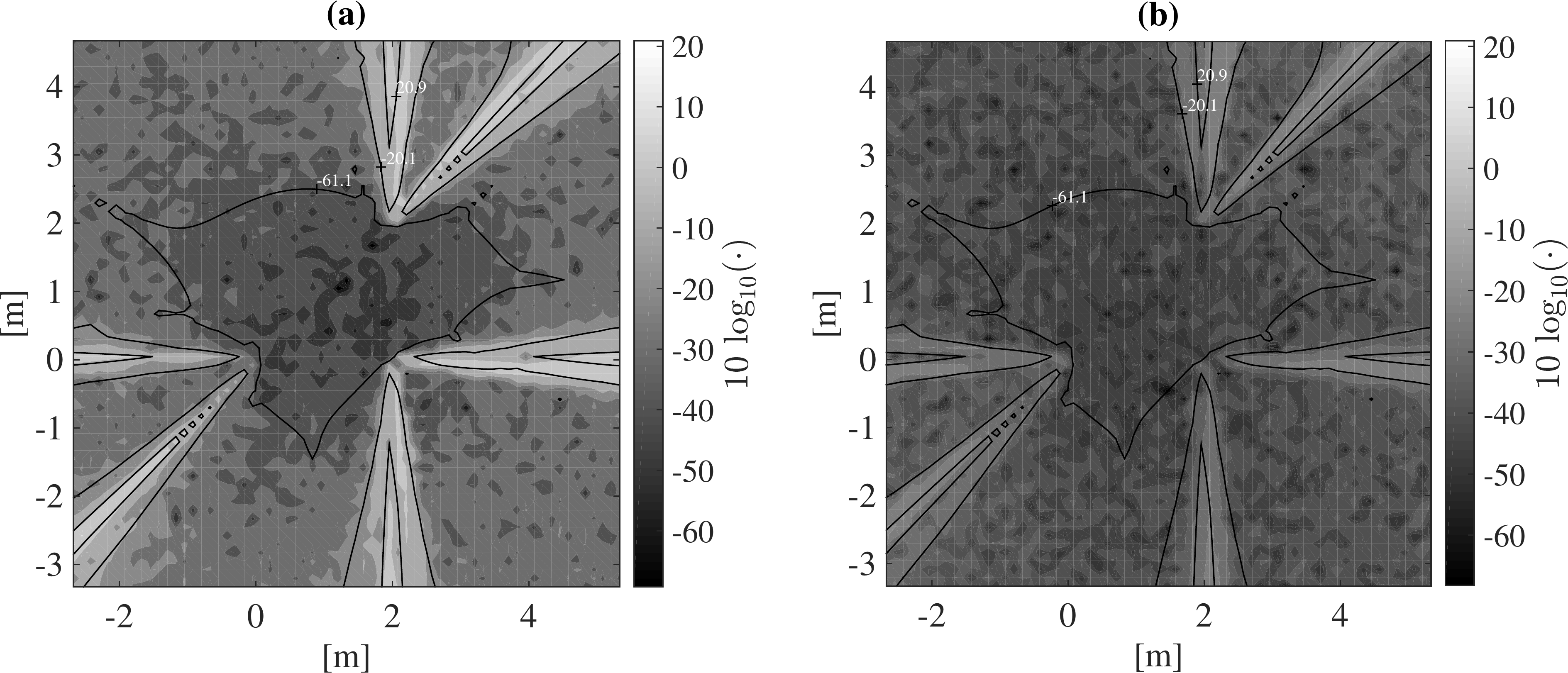}}
  \caption{\label{fig:rem_analysis}Superposition of the contour lines of $\bs{\Delta}(\x)$ (solid lines), and $\bs{\hat\Delta}(\x)$ (contour with colormap). We plot their projections along the eigenvectors related to the maximum (a) and minimum (b) eigenvalue of $\mb{G}^{-1}(\mathbf{x}),$ respectively. Results are expressed on a logarithmic scale.}
\end{center}
\end{figure}

Similarly to RMSE, Figure \ref{fig:bias_analysis} illustrates the bias analysis.
The first and second rows show asymptotic prediction $-\frac{1}{2}\mb{b}(\x)$ and the sample estimate $\mb{\hat B}(\x)$, respectively. As usual, we project them along the maximum (a) and minimum (b) eigenvectors of $\mb{G}^{-1}(\mathbf{x})$. We observe again a good match between simulations and prediction, except for the areas where the remainder analysis suggests a relevant error. It is worth noticing that at the points that are far from these areas the bias exhibits very small magnitude, which makes it difficult to estimate its value from simulations. For this reason, the contour lines of expected and simulated values take on different shapes. 
\begin{figure}[htb]
\begin{center}
\subfigure{\resizebox{15cm}{!}{
\includegraphics*[bb=0 0 1217 592]{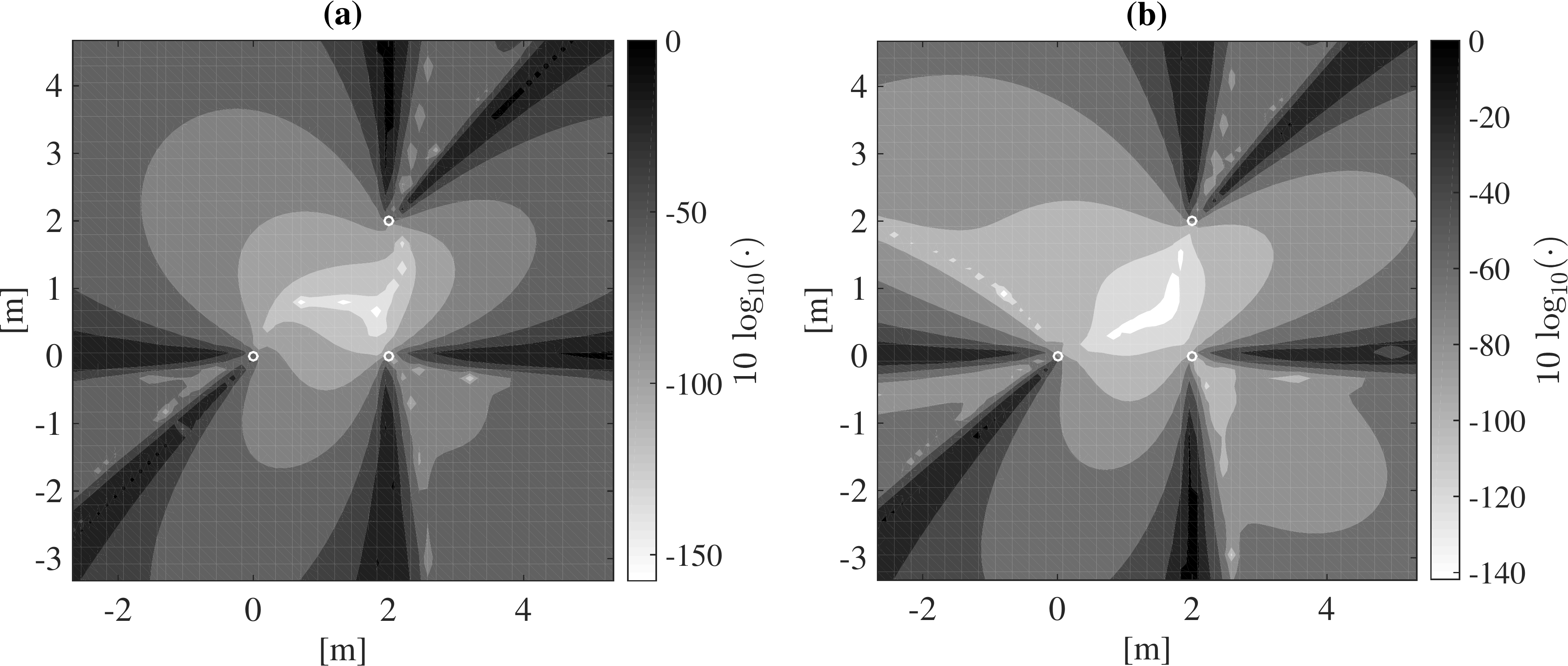}}}\\[5mm]
\subfigure{\resizebox{15cm}{!}{
\includegraphics*[bb=0 0 1217 592]{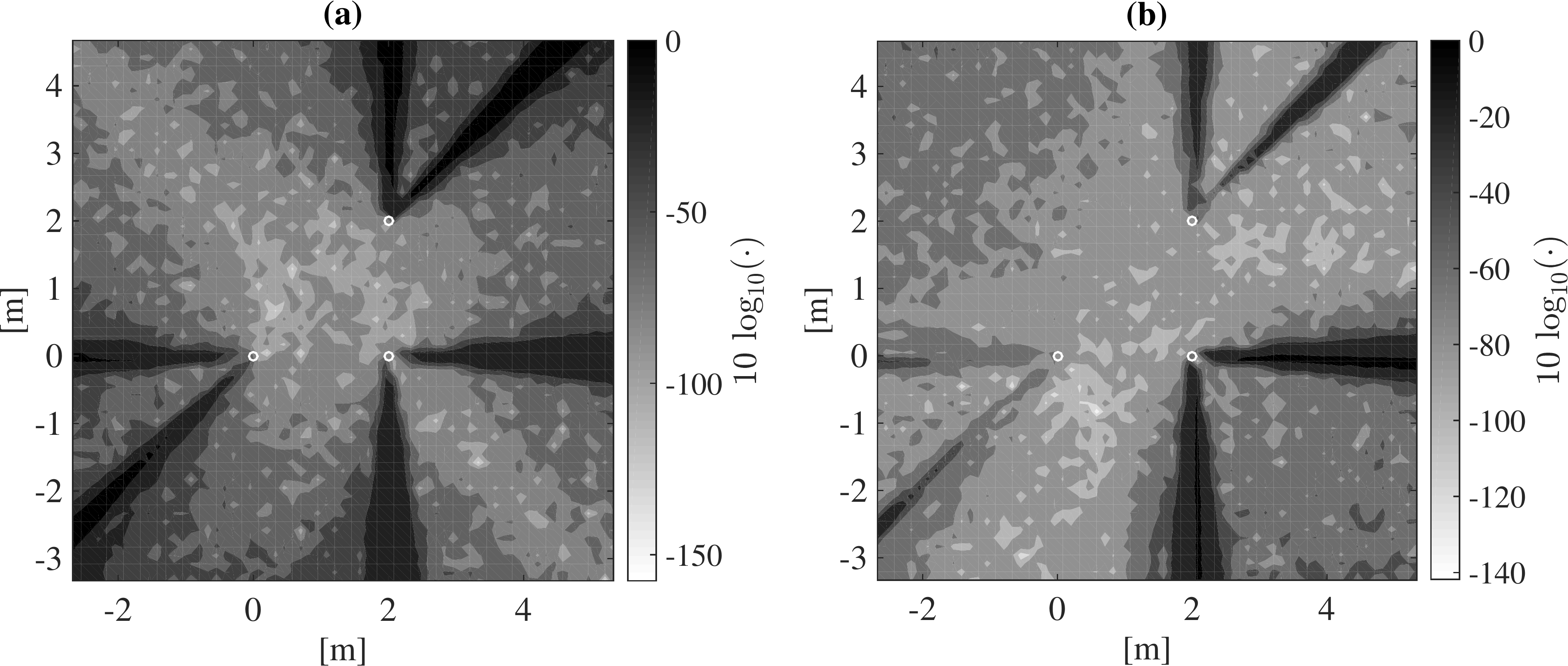}}}
\caption{\label{fig:bias_analysis}First row: asymptotic prediction $-\frac{1}{2}\mb{b}(\x)$ according to eq.~\eqref{eq:ABMLE}. Second row: sample estimate $\mb{\hat B}(\x)$ of the bias from simulations. Projections along the eigenvectors related to the maximum (a) and minimum (b) eigenvalue of $\mb{G}^{-1}(\mathbf{x})$ are shown. Results are expressed on a logarithmic scale.}
\end{center}
\end{figure}
\subsection{MLE in real scenarios}\label{sec:blindMLE}
The analysis we carried out in the previous sections is based on the decomposition of the TDOA--based localization into four different models. However, real world situations are quite different. Unless we have some a-priori information on the source position (e.g. in indoor localization the source lies on a bounded subset of $\RR^2$), it is impossible to know which model to use, or equivalently, on which region of the $x$--plane the source lies. This means that, in general, it is not possible to avoid the ambiguity of the source localization described in Section \ref{sc:TDOAspace}. This is what we call the \emph{blind} localization scenario.


Let us consider the model $M.$ We remind that it is a $(2,2)$ curved exponential family whose parameter space $U$ is a subset of $\Theta=\RR^2.$ Given measurements $\bs{\hat\tau}\in\RR^2$, we want to test the null hypothesis $H_0$ that $\bs{\hat\tau}$ is not an outlier for $M$ versus the alternative $H_1$ that $\bs{\hat\tau}$ is an outlier. If this null hypothesis holds valid, then $\x_+(\bs{\bar\tau}(\bs{\hat\tau};\bs{\Sigma_2}))$ should be seen as an admissible source position, otherwise it should not. The likelihood ratio test statistic is
$$
\mu^*(\bs{\hat\tau};\bs{\Sigma_2})=
\frac{\displaystyle\sup_{\bt\in U}\,
l(\bt;\bs{\hat\tau},\bs{\Sigma_2})}
{\displaystyle\sup_{\bt\in\Theta\setminus U}\,l(\bt;\bs{\hat\tau},\bs{\Sigma_2})}=
\exp\left[-\frac{1}{2}\,
\Vert\bs{\hat\tau}-\bs{\bar\tau}(\bs{\hat\tau};\bs{\Sigma_2}))\Vert_{\bs{\Sigma_2}^{\unaryminus 1}}^2\right].
$$
The logarithm of $\mu^*(\bs{\hat\tau};\bs{\Sigma_2})$ follows a Chi--square distribution with two degrees of freedom:
$$
-2\ln\mu^*(\bs{\hat\tau};\bs{\Sigma_2})=\Vert\bs{\hat\tau}-\bs{\bar\tau}(\bs{\hat\tau};\bs{\Sigma_2})\Vert_{\bs{\Sigma_2}^{\unaryminus 1}}^2\sim\mathcal{X}_1^2\,.
$$
The cutoff between significant and non-significant results depends on the application, but a typical level is $0.05$. This means that the critical region, i.e. the set of $\bs{\hat\tau}$  for which the null hypothesis is rejected, is approximately
$$
\mathcal{R}_{\bs{\Sigma_2}}=
\{\bs{\hat\tau}\in\Theta\setminus U\,|\,\Vert\bs{\hat\tau}-\bs{\bar\tau}(\bs{\hat\tau};\bs{\Sigma_2})\Vert_{\bs{\Sigma_2}^{\unaryminus 1}}^2
\geq 3.84\}.
$$
A similar reasoning can be followed for the models $M_0,M_1,M_2$. In Figure \ref{fig:modelliLRT} we draw the four models and their acceptance regions (the complement in $\Theta$ of the rejected ones).
\begin{figure}[H]
\begin{center}
\resizebox{5cm}{!}{
  \includegraphics[bb=226 296 385 495]{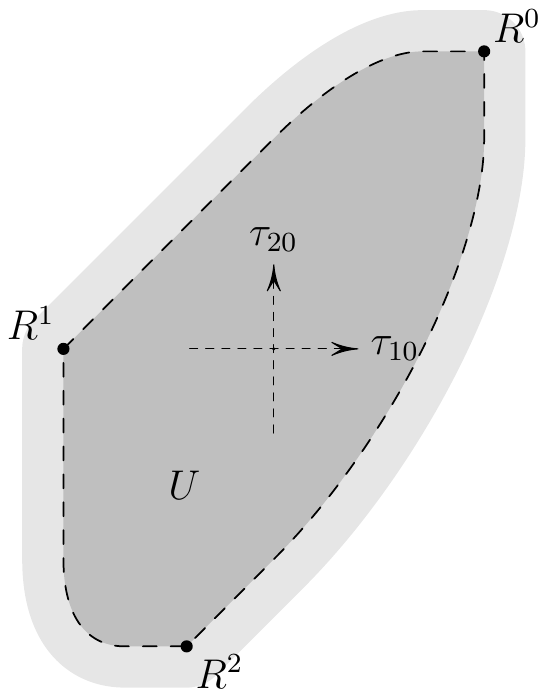}}
  \hspace{20mm}
\resizebox{5cm}{!}{
  \includegraphics[bb=227 296 384 495]{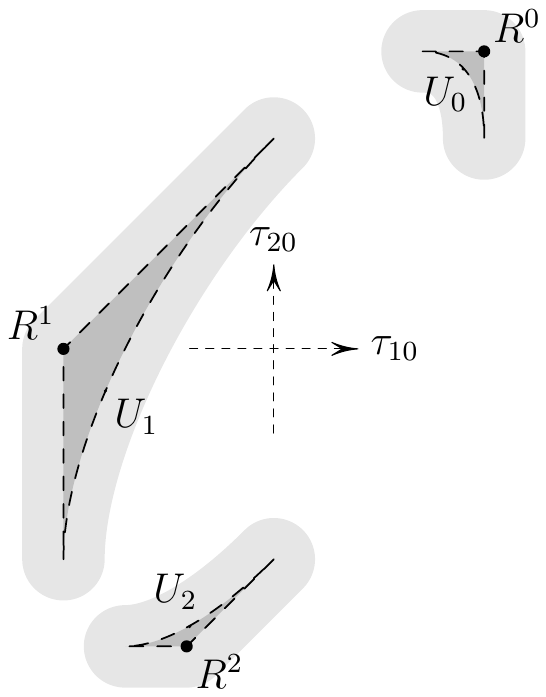}}\\
  \caption{On the left, the medium gray subset is the feasible set $U$ of the model $M$. On the right, the medium gray subset $U_i$ is the feasible set of the model $M_i,\ i=0,1,2.$ The light gray subsets are the regions of acceptance of the measurements for the various models, for the level $0.05$ and $\bs{\Sigma_2}=\sigma^2\bs I,$ where $\sigma=0.1\,d_{10}$. Therefore, two TDOAs defining a point $\bs{\hat{\tau}}=(\hat\tau_{10},\hat\tau_{20})$ outside a certain gray region should be considered outliers with respect to the corresponding model.}
  \label{fig:modelliLRT}
\end{center}
\end{figure}

We summarize the above discussion in the following algorithm.
\begin{algorithm}[H]
\caption{MLE algorithm for blind localization}
\label{alg:MLE_blind}
\begin{algorithmic}[1]
\REQUIRE{TDOA measurements $\bs{\hat\tau}\in\RR^2,$ covariance matrix $\bs{\Sigma_2}$}
\STATE{Compute $\bs{\bar\tau}(\bs{\hat\tau};\bs{\Sigma_2})$ and $\bs{\bar\tau_i}(\bs{\hat\tau};\bs{\Sigma_2}),\ i=0,1,2.$}
\STATE\label{line2}{If $\Vert\bs{\hat\tau}-\bs{\bar\tau}(\bs{\hat\tau};\bs{\Sigma_2})\Vert_{\bs{\Sigma_2}^{\unaryminus 1}}^2<3.84$ accept the estimation $\bs{\bar\tau}(\bs{\hat\tau};\bs{\Sigma_2}),$ else discard the measurements $\bs{\hat\tau}$. The same for $\bs{\bar\tau_i}(\bs{\hat\tau};\bs{\Sigma_2}),\ i=0,1,2.$}
\STATE{The admissible source positions are $\mb{\bar{x}}(\bs{\hat\tau};\bs{\Sigma_2})=\x_+(\bs{\bar\tau}(\bs{\hat\tau};\bs{\Sigma_2}))$ and $\mb{\bar{x}_i}(\bs{\hat\tau};\bs{\Sigma_2})=\x_-(\bs{\bar\tau_i}(\bs{\hat\tau};\bs{\Sigma_2})),\ i=0,1,2,$ among the ones passing the likelihood ratio test in Line \ref{line2}.}
\RETURN $\mb{\bar{x}}(\bs{\hat\tau};\bs{\Sigma_2})$ and $\mb{\bar{x}_i}(\bs{\hat\tau};\bs{\Sigma_2})$
\end{algorithmic}
\end{algorithm}
Finally, for each admissible source position we have the associated error estimation as computed in Section \ref{sc:ASI} and validated in this section.

\section{Conclusion and perspective}\label{sec:conclusion}
In this manuscript we studied the statistical model for the Range Difference--based localization, in the minimal scenario of three sensors and a source lying in the plane of the sensors. This analysis arises from previous works \cite{Compagnoni2013b,Compagnoni2013a}, where the deterministic model for source localization was described in terms of mapping from the physical space into the measurement space (TDOA space).

First of all, we faced the problem of multiple source location by defining four distinct curved exponential families, one for each region of the physical plane where the TDOA map is injective. On this basis, we developed a Maximum Likelihood technique for localization, which works in the TDOA space and is formulated in a closed-form. Then, we obtained the asymptotical inference of the root mean square error and bias on source position using tools of Information Geometry. We showed via simulation the quality and accuracy of such predictions. Moreover, we investigated higher order statistics in order to evaluate the reliability of the asymptotic analysis. These instruments are fundamental for treating relevant problems in applications, such as the optimal placements of sensors. These aspects are currently under further investigation.

We conducted this analysis under the assumption of knowing in advance the statistical model to use, which is not always applicable in a real scenario. In Section \ref{sec:blindMLE} we proposed an algorithm for source localization that applies to such situations, which is based on the likelihood ratio test. Following the same approach, we are currently developing an outlier removal technique \cite{Compagnoni2016c}, which is strongly based on the description of the statistical model given in this manuscript.

\section*{Acknowledgments}
The authors would like to thank Alessandra Guglielmi for useful discussions and suggestions during the preparation of this work.


\appendix

\renewcommand{\thesection}{A}
\section{The Mahalanobis degree discriminant}\label{app:evoluta}
By eliminating $\nu$ from system \eqref{eq:MDequation} we get the following algebraic equations
\begin{equation}\label{eq:MDequation2}
\left\{\begin{array}{l}
\langle\bs{\hat\tau}-\bt,\mathbf{H}\,\nabla_{\bt}a(\bt)\rangle_{\bs\Sigma_2^{\unaryminus 1}}=0\\
a(\bt)=0
\end{array}\right..
\end{equation}
By Definition \ref{def:MDdegre}, the Mahalanobis degree discriminant $\mathcal{E}_{\bs{\Sigma}_2}$ is the locus of the points $\bs{\hat\tau}$ of the $\tau$--plane where \eqref{eq:MDequation2} has at least two coinciding solutions. This condition is equivalent to require that $\nabla_{\bt}a(\bt)$ and $\nabla_{\bt}\langle\bs{\hat\tau}-\bt,\mathbf{H}\,\nabla_{\bt}a(\bt)\rangle_{\bs\Sigma_2^{\unaryminus 1}}$ are parallel.
A Singular code \cite{DGPS} for computing $\mathcal{E}_{\bs{\Sigma}_2}$ is:
\begin{verbatim}
ring r=0,(t1,t2,u,v),lp;
LIB"linalg.lib";
matrix d10[2][1]=d10x, d10y;
matrix d20[2][1]=d20x, d20y;
matrix S[2][2]=s11,s12,s21,s22;
matrix IS=inverse(S);
matrix H[2][2]=0,-1,1,0;
poly a=(d10[1,1]^2+d10[2,1]^2)*t2^2-2*(d10[1,1]*d20[1,1]+d10[2,1]*d20[2,1])*t1*t2
+(d20[1,1]^2+d20[2,1]^2)*t1^2-(d10[1,1]*d20[2,1]-d20[1,1]*d10[2,1])^2;
matrix GA[2][1]=diff(a,t1),diff(a,t2);
matrix NN=-H*IS*H*GA;
matrix MR[2][2]=u-t1, v-t2, NN[1,1], NN[2,1];
poly r=det(MR);
matrix MT[2][2]=diff(R,t1), diff(R,t2), GA[1,1], GA[2,1];
poly t=det(MT);
ideal i=a,r,t;
i=std(i);
print(i[1]);
\end{verbatim}
The Cartesian equation of $\mathcal{E}_{\bs{\Sigma}_2}$ is the degree-$6$ polynomial in $u,v$ that is printed thanks to the last line. The only data to be changed are the numerical values of the matrices \verb|d10| and \verb|d20| equal to the components of the displacement vectors and the covariance matrix \verb|S|.


\bibliographystyle{hplain}   
\bibliography{biblio}       

\end{document}